\documentclass{cccg12}
\usepackage{graphicx,amssymb,amsmath,url}

\newcommand{\lemlab}[1]{\label{lemma:#1}}
\newcommand{\thmlab}[1]{\label{thm:#1}}

\newcommand{\corlab}[1]{\label{cor:#1}}

\newcommand{\figlab}[1]{\label{fig:#1}}
\newcommand{\seclab}[1]{\label{sec:#1}}

\newcommand{\lemref}[1]{\ref{lemma:#1}}
\newcommand{\thmref}[1]{\ref{thm:#1}}
\newcommand{\corref}[1]{\ref{cor:#1}}

\newcommand{\secref}[1]{\ref{sec:#1}}

\newcommand{\figref}[1]{\ref{fig:#1}}

\def\P{{\mathcal P}}
\def\kap{{\kappa}}

\def\a{{\alpha}}
\def\b{{\beta}}

\title{Unfolding Prismatoids as Convex Patches:\\
Counterexamples and Positive Results}

\author{%
Joseph O'Rourke%
    \thanks{Department of Computer Science, Smith College, Northampton, MA
      01063, USA.
      {\tt orourke@cs.smith.edu}.
This paper was prepared for but never submitted to \emph{CCCG'12}.
It still retains that conference's formatting conventions.
      }
}

\index{O'Rourke, Joseph}

\begin{document}
\thispagestyle{empty}
\maketitle

\begin{abstract}
We address the unsolved problem of unfolding prismatoids
in a new context, viewing  a ``topless prismatoid''
as a \emph{convex patch}---a polyhedral subset of the surface of a convex
polyhedron
homeomorphic to a disk.
We show that several natural strategies for unfolding a prismatoid can
fail,
but obtain a positive result for ``petal unfolding" topless prismatoids.
We also show that the natural extension 
to a convex patch consisting of a
face of a polyhedron
and all its incident faces, does not always have a nonoverlapping petal unfolding.
However, we
obtain a positive result by excluding the problematical patches.
This then leads a positive result for restricted prismatoids.
Finally, we suggest suggest studying the unfolding of
convex patches in general, and offer some possible lines of investigation.
\end{abstract}

\section{Introduction}
\seclab{Introduction}
A \emph{prismatoid} is the convex hull of two convex polygons $A$
(above) and $B$ (base) in
parallel planes.  Despite its simple structure, it remains unknown
whether or not every prismatoid has a nonoverlapping edge unfolding, a narrow
special case of what has become known as D\"urer's Problem: whether
every convex polyhedron has a nonoverlapping edge unfolding~\cite[Prob.~21,1,
p.~300]{do-gfalop-07}.
(All polyhedra considered here are convex polyhedra, and we will
henceforth drop the modifier ``convex," and consistently use the symbol $\P$; 
we will also use \emph{unfolding} to mean ``nonoverlapping edge unfolding.")
Motivated by the apparent difficulty of placing the top in an unfolding, we explore
unfolding \emph{topless prismatoids}, those with the top $A$ removed.
We show that several natural approaches fail, but that a somewhat complex algorithm
does succeed in unfolding any topless prismatoid.

This success suggests studying the unfolding of a \emph{convex patch} more generally:
a connected subset of faces of a polyhedron $\P$, homeomorphic to a disk.
A natural convex patch is an extension of a class studied by Pincu~\cite{p-ofnp-07}.
He proved that the patch that consists of one face $B$ of $\P$ and every face 
that shares an edge with $B$, has a ``petal unfolding" (defined below).
This \emph{edge-neighborhood} of a face is itself a natural generalization of ``domes,"
earlier proven to have a petal unfolding~\cite[p.~323ff]{do-gfalop-07}.
The generaliziation we explore is the \emph{vertex-neighborhood} of a face: $B$ together
with every face that shares at least a vertex with $B$.  We show that not every vertex-neighborhood
patch has a petal unfolding.  Note that every topless prismatoid is a vertex-neighborhood of its base.
This negative result suggests a restriction that permits unfolding: if $\P$ is nonobtusely triangulated,
then the vertex-neighborhood of any face does have a petal unfolding.  This in turn leads to a proof
that triangular prismatoids (top included), composed of nonobtuse triangles, have an unfolding.

Finally, we make a few observations and conjectures about unfolding arbitrary convex patches.

\subsection{Band and Petal Unfoldings}
\seclab{BandPetal}
There are two natural unfoldings of a prismatoid.
A \emph{band unfolding} cuts one lateral edge and unfolds all
lateral faces as a unit, called a \emph{band}, leaving $A$ and $B$
attached each by one uncut edge to opposite sides of the band
(see, e.g.,~\cite{adlmost-eunpb-07}).
Aloupis showed that the lateral edge can be chosen so that band alone unfolds~\cite{a-rps-05},
but I showed that, nevertheless, there are prismatoids
such that every band unfolding overlaps~\cite{o-bupc-07}.
The example will be repeated here, as it plays a role in Sec.~\secref{ConvexPatches}.

The prismatoid with no band unfolding is shown in Fig.~\figref{BandedHex3D}.
\begin{figure}[htbp]
\centering
\includegraphics[width=0.8\linewidth]{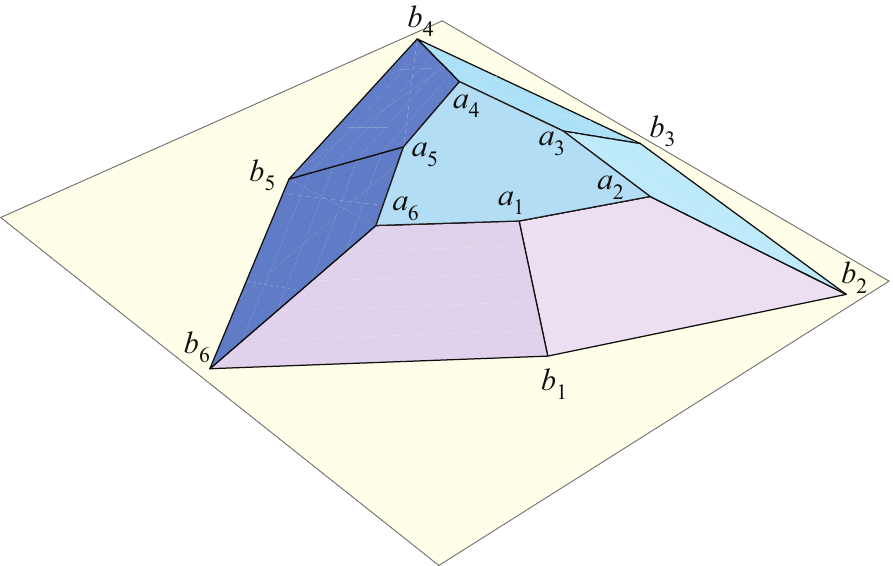}
\caption{The banded hexagon. The curvatures at the three
side vertices $\{a_2,a_4,a_6\}$ is $2^\circ$,
and that at the apex vertices $\{a_1,a_3,a_5\}$ is $7.5^\circ$.}
\figlab{BandedHex3D}
\end{figure}
The possible band unfoldings are shown in the Appendix, Figs.~\figref{ApexCuts} and~\figref{SideCuts}.
Note that this example also establishes that not every topless prismatoid has a band unfolding,
simply by interchanging the roles of $A$ and $B$.

The second natural unfolding is a \emph{petal unfolding}, called a
``volcano unfolding'' 
in~\cite[p.~321]{do-gfalop-07}.
Because Fig.~\figref{BandedHex3D} without its base is a edge-neighborhood patch, it
can be petal-unfolded by Pincu's result~\cite{p-ofnp-07} as noted above: simply
cut each lateral edge $a_i b_i$.

Let $\P$ be a prismatoid, and assume all lateral faces are triangles,
the generic and seemingly most difficult case.
Let $A=(a_1,a_2,\ldots)$
and $B=(b_1,b_2,\ldots)$.
Call a lateral face that shares an edge with $B$ a \emph{base} or
$B$-triangle,
and a lateral face that shares an edge with $A$ a \emph{top} or
$A$-triangle.
A petal unfolding cuts no edge of $B$, and unfolds every base triangle
by rotating it around its $B$-edge into the base plane.
The collection of $A$-triangles incident to the same $b_i$ vertex---the
\emph{$A$-fan} $AF_i$---must
be partitioned into two groups, one of which rotates clockwise (cw) to
join with the unfolded base triangle to its left, and the other group
rotating counterclockwise (ccw) to
join with the unfolded base triangle to its right. Either group could
be empty.  Finally, the top $A$ is attached to one $A$-triangle.
So a petal unfolding has choices for how to arrange the $A$-triangles,
and which $A$-triangle connects to the top.
See Fig.~\figref{TriPrismatoidExample} in the Appendix for an example.

As of this writing, it remains possible that every prismatoid has a petal unfolding:
I have not been able to find a counterexample.
For a hint of why placing the top in a petal unfolding seems problematical, see Fig.~\figref{DrumOverlap} in the Appendix.
The next section presents the main result: every topless prismatoid has a petal unfolding.

\section{Topless Prismatoid Petal Unfolding}
\seclab{ToplessPetalUnfolding}
An example of a petal unfolding of a topless prismatoid is shown in
Fig.~\figref{ExampleUnf58}.
\begin{figure}[htbp]
\centering
\includegraphics[width=0.75\linewidth]{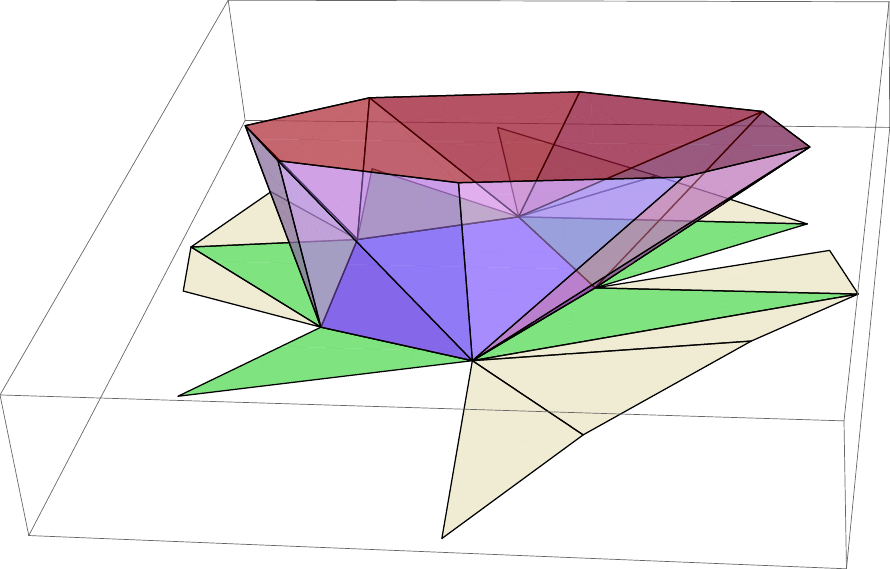}
\caption{Unfolding of a topless prismatoid}
\figlab{ExampleUnf58}
\end{figure}
Even topless prismatoids present challenges.
For example, consider the special case when there is only one $A$-triangle between every two $B$-triangles.
Then the only 
choice for placement of the $A$-triangles is whether to turn each ccw or cw.
It is natural to hope that rotating all $A$-triangles consistently ccw or cw suffices to avoid
overlap, but this can fail, as in Fig.~\figref{DrumOverlap}, and even for triangular prismatoids,
Fig.~\figref{WingsCrossCcw} in the Appendix.
A more nuanced approach would turn each $A$-triangle
so that its (at most one) obtuse angle is not joined to a $B$-triangle (resolving Fig.~\figref{WingsCrossCcw}), but this can fail also,
a claim I will not substantiate.

The proof follows this outline:
\begin{enumerate}
\item An ``altitudes partition'' of the plane exterior to the \emph{base unfolding} ($B$ plus all $B$-triangles) is defined
and proved to be a paritition.
\item It is shown that both $\P$ and
this partition vary in a natural manner with respect to the separation $z$ between the $A$- and $B$-planes.
\item An algorithm is detailed for petal unfolding the $A$-triangles for the ``flat prismatoid'' $\P(0)$, the
limit of $\P(z)$ as $z \to 0$, such that these $A$-triangles fit inside the regions of the altitude partition.
\item It is proved that nesting within the partition regions remains true for all $z$.
\end{enumerate}

\subsection{Altitude Partition}
\seclab{AltitudePartition}
We use $a_i$ and $b_j$ to represent the vertices of $\P$, and primes to indicate unfolded images on
the base plane.

Let $B_i = \triangle b_i b_{i+1} a'_j$ be the $i$-th base triangle.
Say that $BU=B \cup (\bigcup_i B_i)$ is the \emph{base unfolding}, the unfolding
of all the $B$-triangles arrayed around $B$ in the plane, without any $A$-triangles.
The altitude partition partitions the plane exterior to the base unfolding.

Let $r_i$ be the \emph{altitude ray} from $a'_j$ along the altitude of $B_i$.
Finally, define $R_i$ to be the region of the plane incident to $b_i$,
including the edges of the $B_{i-1}$ and $B_i$ triangles incident to $b_i$,
and bounded by $r_{i-1}$ and $r_i$.
See Fig.~\figref{AltitudeRays1}.
\begin{figure}[htbp]
\centering
\includegraphics[width=0.75\linewidth]{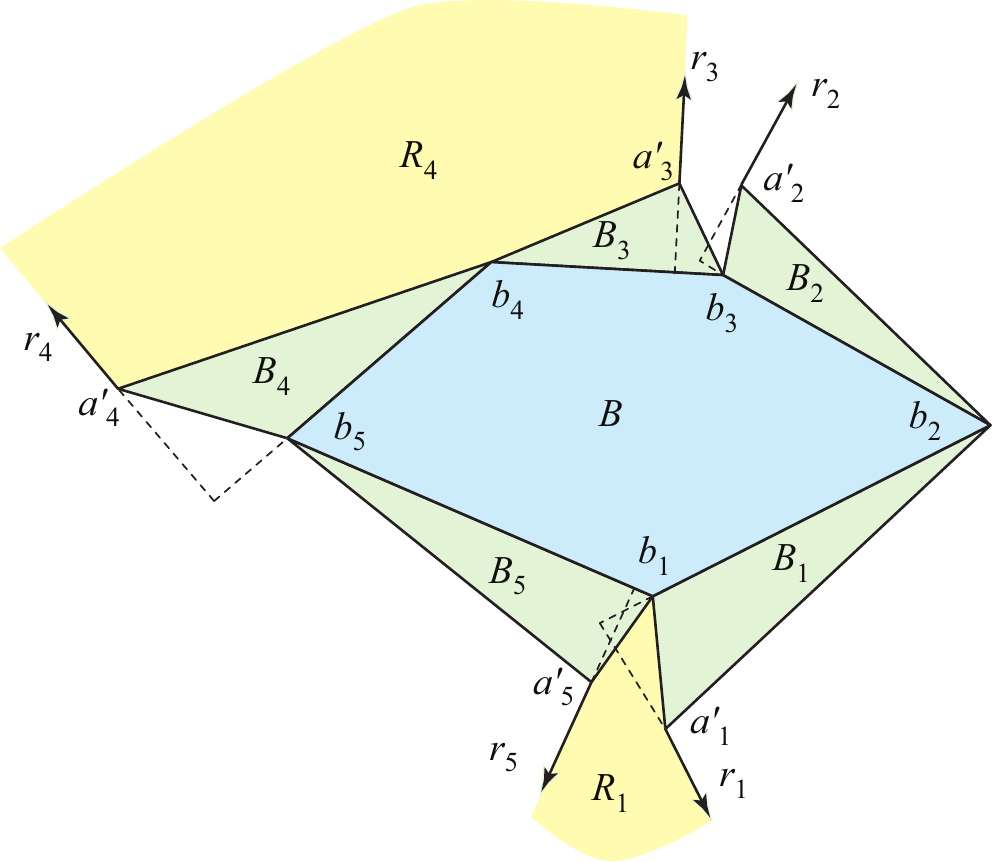}
\caption{Partition exterior to the base unfolding by altitude rays $r_i$. 
In this example both $A$ and $B$ are pentagons;
in general there would not be synchronization between the $b_i$ and $a_i$ indices.
The $A$-triangles are not shown.}
\figlab{AltitudeRays1}
\end{figure}

\begin{lemma}
No pair of altitude rays cross in the base plane, and so they define a partition of that plane exterior
to the base unfolding $BU$.
\lemlab{AltitudeRays}
\end{lemma}
\begin{proof}
See Sec.~\secref{ProofAltitudeRays} in the Appendix.
\end{proof}

Our goal is to show that the $A$-fan $AF_i$ incident to $b_i$ can be partitioned into two groups,
one rotated cw, one ccw, so that both fit inside $R_i$.
(Note that this nesting is violated in Fig.~\figref{WingsCrossCcw} in the Appendix.)

\subsection{Behavior of $\P(z)$}
\seclab{BehaviorPz}

We will use ``$(z)$" to indicate that a quantity varies with respect to the height $z$ separating
the $A$- and $B$-planes.

\begin{lemma}
Let $\P(z)$ be a prismatoid with height $z$.
Then the combinatorial structure of $\P(z)$ is independent of $z$,
i.e., raising or lowering $A$ above $B$ retains the convex hull structure.
\lemlab{zto0}
\end{lemma}
\begin{proof}
See Sec.~\secref{ProofAltitudeRays} in the Appendix.
\end{proof}

We will call $\P(0) = \lim_{z \to 0} \P(z)$ a \emph{flat prismatoid}.
Each lateral face of $\P(0)$ is either an \emph{up-face} or a \emph{down-face},
and the faces of $\P(z)$ retain this classification in that their outward normals
either have a positive or a negative vertical component.

\begin{lemma}
Let $\P(z)$ be a prismatoid with height $z$, and $BU(z)$ its base unfolding.
Then the apex $a'_j(z)$ of each $B'_i(z)$ triangle $\triangle b_i b_{i+1} a'_j(z)$
in $BU(z)$ lies on the
fixed line containing the altitude of $B'_i(z)$.
\lemlab{AltitudeRaysUnf}
\end{lemma}
\begin{proof}
See Sec.~\secref{AltitudeRaysUnf} in the Appendix.
\end{proof}

Thus the vertices $a'_j(z)$ of the base unfolding
``ride out" along the altitude rays $r_i$ as $z$ increases (see ahead to Fig.~\figref{FlippingCR}
for an illustration).
Therefore the combinatorial structure of the altitude partition is fixed,
and $R_i$ only changes geometrically by the lengthening of the edges
$b_i a'_j$ and $b_{i+1} a'_j$ and the change in the angle gap $\kap_{b_i}(z)$ at $b_i$.

\subsection{Structure of $A$-fans}
\seclab{A-fan}
Henceforth we concentrate on one $A$-fan, which we always take to be
incident to $b_2$, and so between $B_1=\triangle b_1 b_2 a_1$
and $B_2=\triangle b_2 b_3 a_k$.
The \emph{$a$-chain} is the chain of vertices $a_1, \ldots, a_k$.
Note that the plane containing $B_1$ supports $A$ at $a_1$, and
the plane containing $B_2$ supports $A$ at $a_k$.
Let $\b=\b_2$ be the base angle at $b_2$: $\b=\angle b_1 b_2 b_3$.
We state here a few facts true of every $A$-fan.

\begin{enumerate}
\item An $a$-chain spans at most ``half" of $A$, i.e., a portion between parallel
supporting lines (because $\b > 0$).
\item If an $A$-fan is unfolded as a unit to the base plane, the $a$-chain
consists of a convex portion followed by a reflex 
followed by a convex portion,
where any of these portions may be empty.
In other words, excluding the first and last vertices, the interior vertices of the chain
have convex angles, then reflex, then convex.
\item Correspondingly, an $A$-fan consists of down-faces followed by up-faces
followed by down-faces, where again any (or all) of these three portions could be empty.
\item All four possible combinations of up/down are possible for the $B_1$ and $B_2$ triangles.
\end{enumerate}
The second fact above is not so easy to see;
its proof is hinted at in Sec.~\secref{CRz} in the Appendix.
The intuition is that there is a limited amount of variation possible in
an $a$-chain. 
It is the third fact that we will use essentially; it will become clear shortly.

\subsection{Flat Prismatoid Case Analysis}
\seclab{FlatCaseAnalysis}
How the $A$-fan is proved to sit inside its altitude region $R$ for $\P(0)$
depends primarly on where $b_2$ sits with respect to $A$, and secondarily
on the three $B$-vertices $(b_1,b_2,b_3)$.
Fig.~\figref{FlatCase1b} illustrates one of the easiest cases, when $b_2$ is
in $C$, the convex region bounded by the $a$-chain and extensions of its
extreme edges.  Then all the $A$-faces are down-faces, the $a$-chain is convex,
one of the two $B$-faces is a down-face ($B_2$ in the illustration), and we
simply leave the $A$-fan attached to that $B$ down-face.
\begin{figure}[htbp]
\centering
\includegraphics[width=0.5\linewidth]{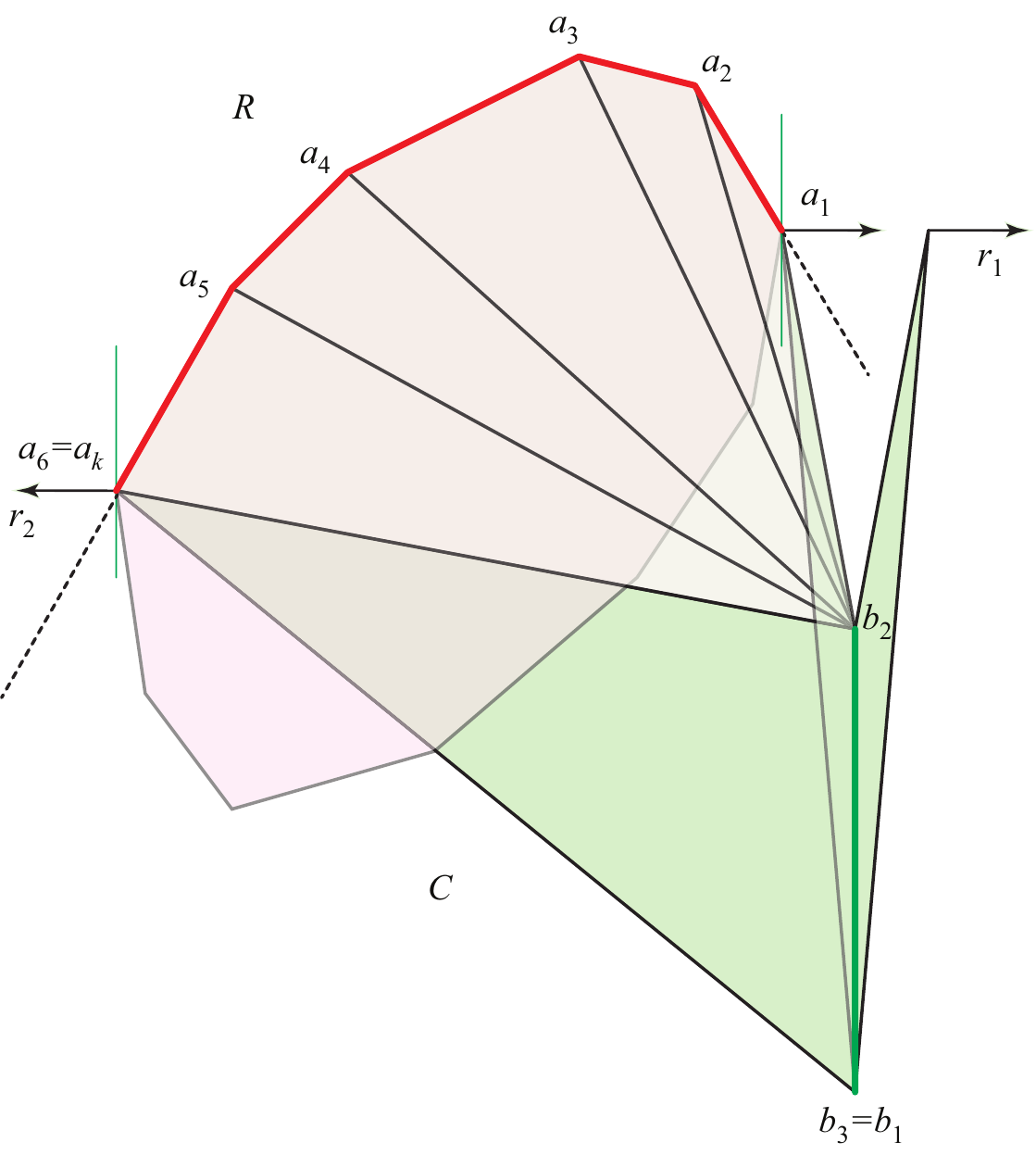}
\caption{Case~1b.  Here we have illustrated  $b_1=b_3$ to allow for
the maximum $a$-chain extent.}
\figlab{FlatCase1b}
\end{figure}

A second case occurs when $b_2$ is on the reflex side of $A$.
An instance when both $B$-triangles are down-faces is illustrated in
Fig.~\figref{FlatFlippingCase2}.  Now the $A$-fan consists of down-faces
and up-faces, the up-faces incident to the reflex side of the $a$-chain.
These up-faces must be flipped in the unfolding, reflected across one of the two tangents
from $b_2$ to $A$. A key point is that not always will
both flips be ``safe'' in the sense that they stay inside the altitude region.
An unsafe flip is illustrated in Fig.~\figref{FlatFlippingCase2Bad}
in the Appendix.
\begin{figure}[htbp]
\centering
\includegraphics[width=\linewidth]{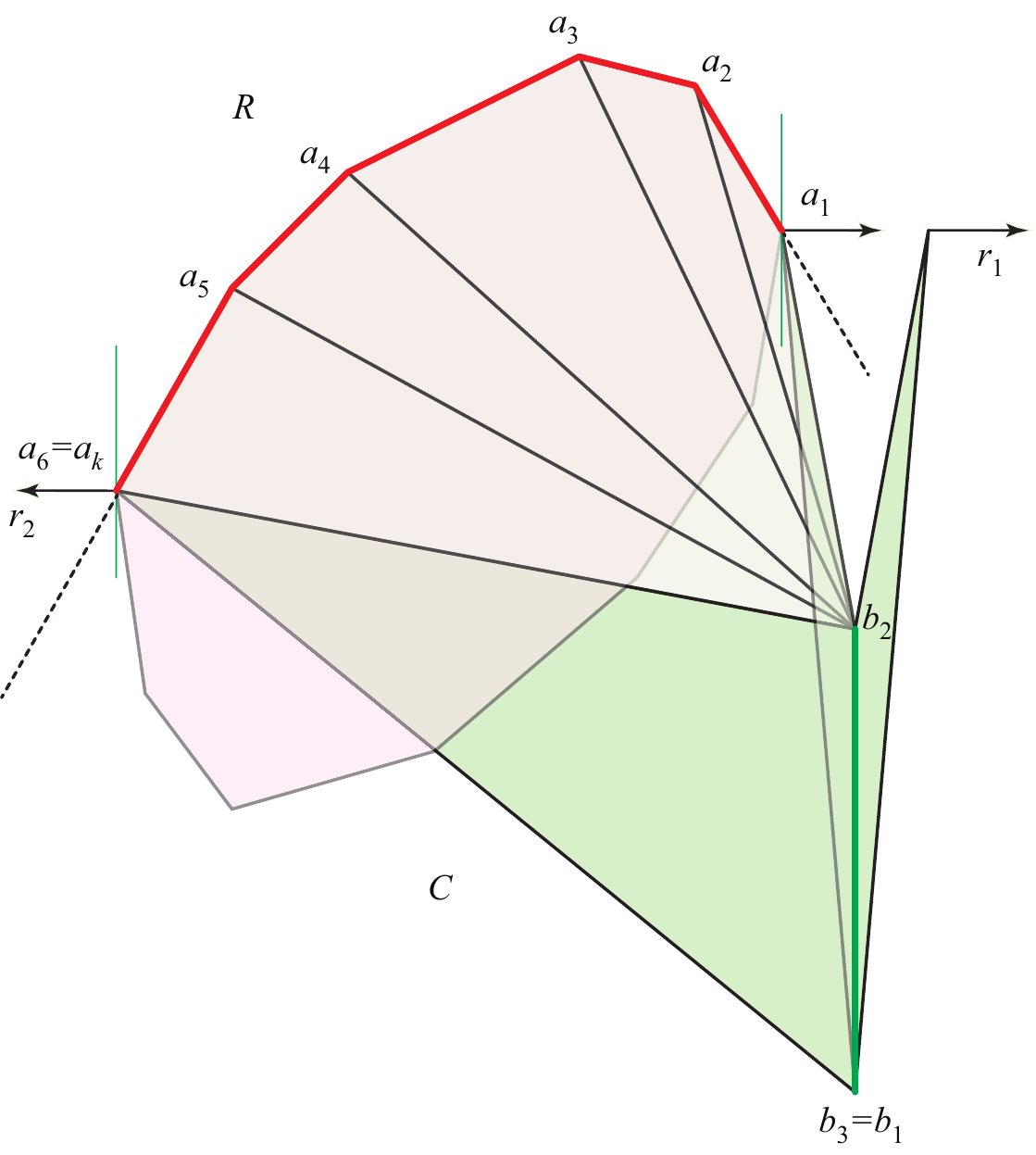}
\caption{Case~2a. The $A$-triangles between the tangents $b_2$ to $a_3$
and $b_2$ to $a_6$ are up-faces. (a)~shows the up-faces flipped over
the left tangent $b_2 a_6$, and (b)~when flipped over the right tangent
$b_2 a_3$.}
\figlab{FlatFlippingCase2}
\end{figure}
Fortunately, one of the two flips is always safe:
\begin{lemma}
Let $b_2$ have tangents touching $a_s$ and $a_t$ of $A$.
Then either reflecting the enclosed up-faces across the left tangent,
or across the right tangent, is ``safe" in the sense that no points of
a flipped triangle falls outside
the rays $r_1$ or $r_k$.
\lemlab{Flip}
\end{lemma}
\begin{proof}
See Sec.~\secref{Flip} in the Appendix.
\end{proof}
The remaining cases are minor variations on those illustrated, and will not be further detailed. See Fig.~\figref{FlatFlippingCase3} in the Appendix.

\subsection{Nesting in $\P(z)$ regions}
\seclab{NestingPz}

The most difficult part of the proof is showing that the nesting established
above for $\P(0)$ holds for $\P(z)$.
A key technical lemma is this:
\begin{lemma}
Let $\triangle b ,a_1(z) ,a_2(z)$ be an $A$-triangle, with angles
$\a_1(z)$ and $\a_2(z)$ at $a_1(z)$ and $a_2(z)$ respectively.
Then $\a_1(z)$ and $\a_2(z)$ are monotonic from their $z=0$ values
toward $\pi/2$ as $z \to \infty$.
\lemlab{AngleMonotonicity}
\end{lemma}
\begin{proof}
See Sec.~\secref{AngleMonotonicity} in the Appendix.
\end{proof}

I should note that it is not true, as one might hope, that the apex angle at $b$ of that $A$-triangle, $\angle  a_1(z), b, a_2(z)$,
shrinks monotonically with increasing $z$, even though its limit as $z \to \infty$ is zero.
Nor is the angle gap $\kap_b(z)$ necessarily monotonic.  These nonmonotonic angle variations
complicate the proof.

Another important observation is that the sorting of $b a_i$ edges by length in $\P(0)$ remains
the same for all $\P(z)$, $z>0$.  More precisely,
let $| b a_i | > | b a_j|$ for two lateral edges connecting vertex $b \in B$
to vertices $a_i,a_j \in A$ in $\P(0)$.
Then $| b a_i(z) | > | b a_j(z)|$ remains true for all $\P(z)$, $z>0$
(by reasoning detailed in Lemma~\lemref{CRz}).

For the nesting proof, I will rely on a high-level description, and one difficult instance.
At a high level, each of the convex or reflex sections of the $a$-chain are enclosed
in a triangle, which continues to enclose that portion of the $a$-chain for any $z>0$
(by Fact~1, Sec.~\secref{CRz}).
See Fig.~\figref{AllConvexEncTri} in the Appendix for the convex triangle enclosure.
The reflex enclosure is determined by the tangents from $b_2$ to $A$: 
$\triangle a_s b_2 a_t$.
So then the task is to prove these (at most three) triangles remain within $R(z)$.
Fig.~\figref{FlippingCR} shows a case where there is both a convex and a reflex section.
Were there an additional convex section, it would remain attached to $B_1(z)$ and
would not increase the challenge.
\begin{figure}[htbp]
\centering
\includegraphics[width=\linewidth]{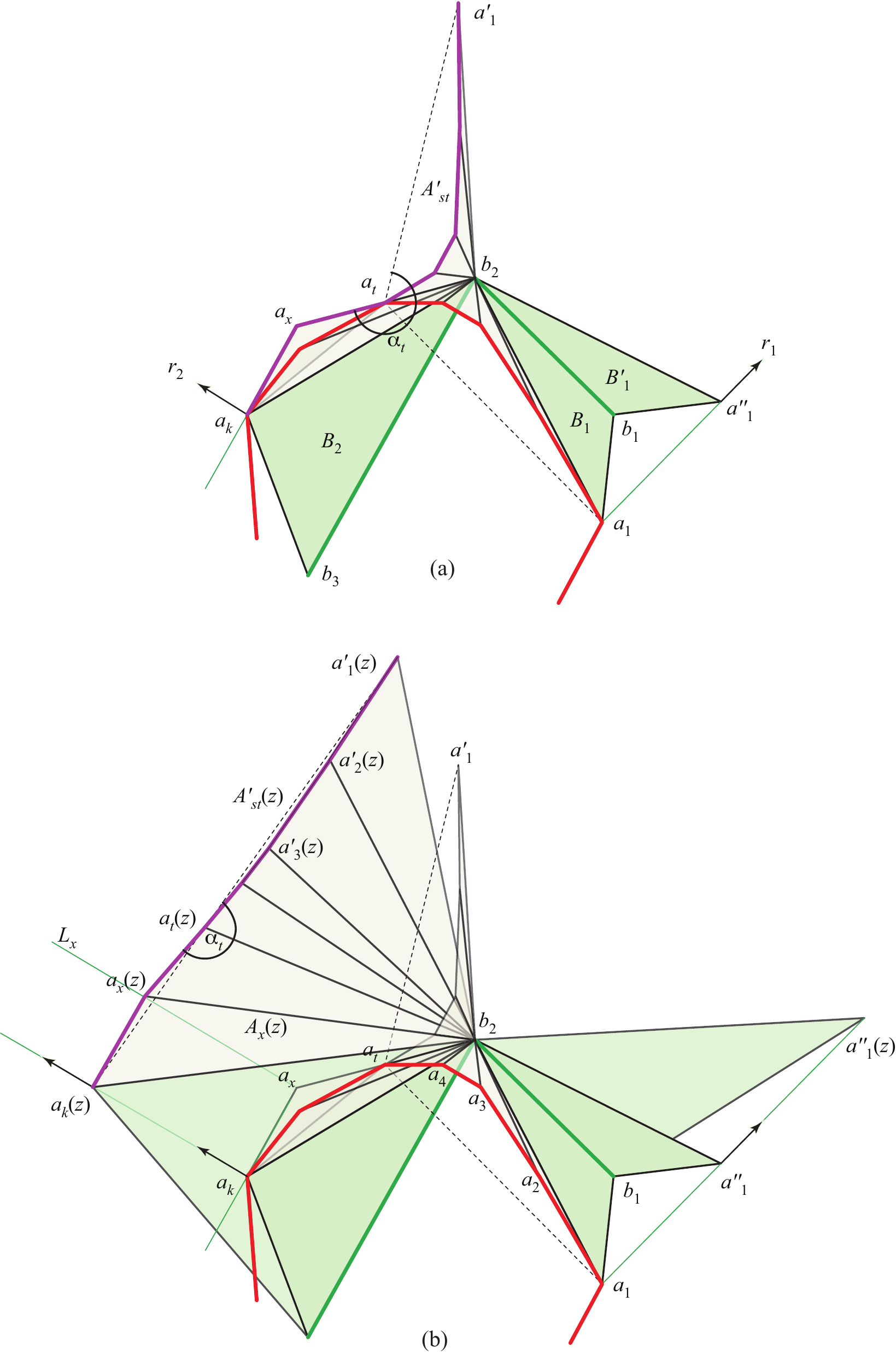}
\caption{(a)~$z=0$. $\triangle a_t a_x a_k$ encloses the convex section,
and $\triangle a_1 b_2 a_t$ encloses the reflex section. (b)~$z > 0$.
Reflex angle $\a_t(z)$ decreases as $z$ increases.}
\figlab{FlippingCR}
\end{figure}

\begin{lemma}
If the $a$-chain consists of a convex and a reflex section, and the safe flip 
(by Lemma~\lemref{Flip}) is to
a side with a down-face ($B_2$ in the figure), then $AF'(z) \subset R(z)$:
the $A$-fan unfolds within the altitude region for all $z$.
\lemlab{CRz}
\end{lemma}
\begin{proof}
See Sec.~\secref{CRz} in the Appendix.
\end{proof}

I have been unsuccessful in unifying the cases in the analysis,
despite their similarity.  Nevertheless, the conclusion is this theorem:
\begin{theorem}
Every triangulated topless prismatoid has a petal unfolding.
\thmlab{ToplessPrismatoids}
\end{theorem}

It is natural to hope that  further analysis will lead to a safe placement of the top $A$
(which might not fit into any altitude-ray region: see Fig.~\figref{DrumOverlap} in the Appendix.

\section{Unfolding Vertex-Neighborhoods}
\seclab{Vertex-Neighborhoods}
Let $N_e(B)$ for a face $B$ of a convex polyhedron $\P$
be $B$ plus the set of all faces that share an edge with $B$,
and $N_v(B)$ be $B$ plus  the set of all faces that share a vertex with $B$.
So $N_v(B) \supseteq N_e(B)$.
As mentioned previously, Pincu proved that $N_e(B)$ has a petal unfolding.
Here we show that $N_v(B)$ does not always have a petal unfolding,
even when all faces in the set are triangles.

A portion of the a 9-vertex example $\P$ that establishes this negative
result is shown in 
Fig.~\figref{CatsEye3D1}.
The $b_1 b_3$ edge of $B$ lies on the horizontal $xy$-plane.
The vertices $\{ b_2, a_1, a_2, c_1, c_2 \}$ all lie on a parallel plane
at height $z$, with $b_2$ directly above the origin: $b_2=(0,0,z)$.

All of $N_v(B)$ is shown in 
Fig.~\figref{CatsEye3D2}.  The structure in Fig.~\figref{CatsEye3D1}
is surrounded by more faces designed to minimize curvatures at the
vertices $b_i$ of $B$.
Finally, $\P$ is the convex hull of the illustrated vertices,
which just adds a quadrilateral ``back'' face $(p_1,c_1,c_2,p_3)$ (not
shown).

The design is such that there is so little rotation possible in the cw and ccw
options for the triangles incident to a vertex of $B$,
that overlap is forced:  see Figs.~\figref{Overlap1}, ~\figref{Overlap2}, and~\figref{Overlap3}.
The thin $\triangle b_2 a_1 a_2$ overlaps in the vicinity of $a_1$ if rotated ccw,
and in the vicinity of $a_2$ is cw (illustrated).
\begin{figure}[htbp]
\begin{minipage}[b]{\linewidth}
\centering
\includegraphics[width=\linewidth]{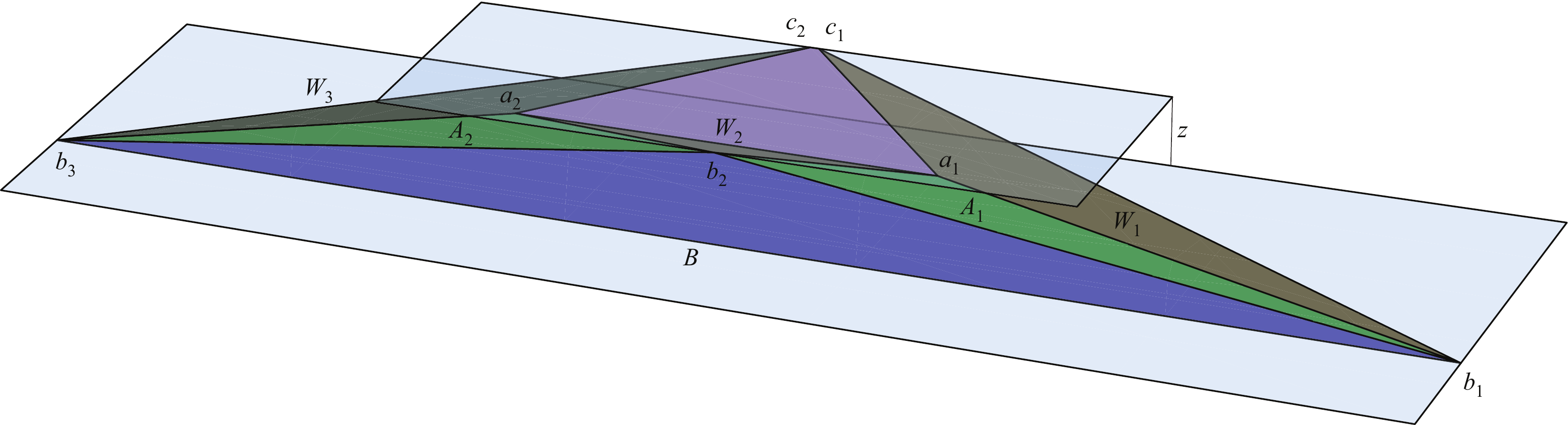}
\caption{Faces of $\P$ in the immediate vicinity of $B$.}
\figlab{CatsEye3D1}
\centering
\includegraphics[width=\linewidth]{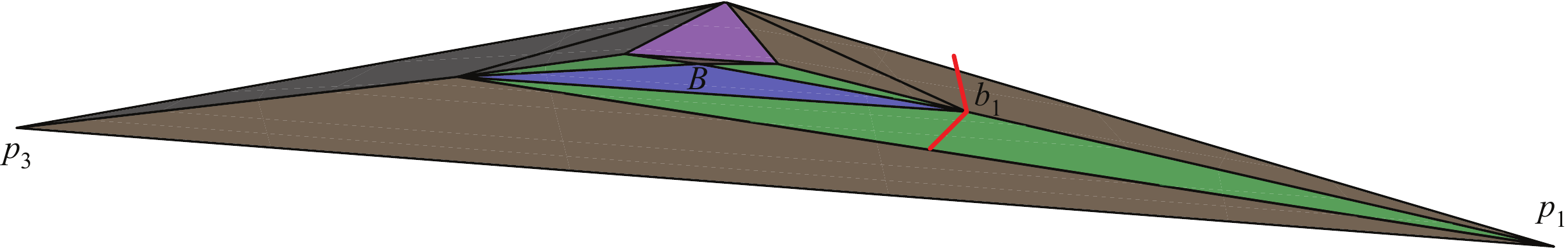}
\caption{All faces incident to $N_v(B)$, and one more, the purple
  quadrilateral $(a_1,c_1,c_2,a_2)$.
The red vectors are normal to $B$ and to $\triangle b_1 p_1 c_1$.}
\figlab{CatsEye3D2}
\centering
\includegraphics[width=\linewidth]{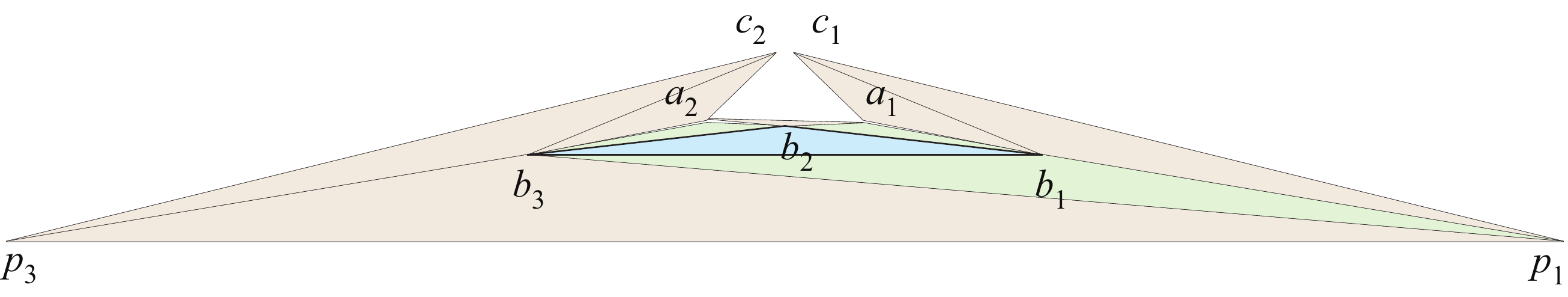}
\caption{Complete unfolding of all faces incident to $B$.}
\figlab{Overlap1}
\centering
\includegraphics[width=0.75\linewidth]{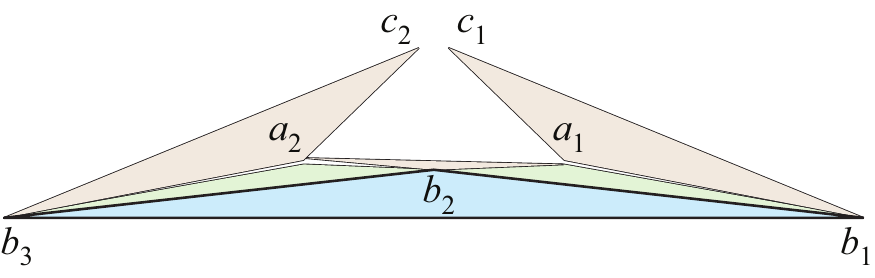}
\caption{Zoom of Fig.~\protect\figref{Overlap1}.}
\figlab{Overlap2}
\centering
\includegraphics[width=0.75\linewidth]{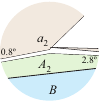}
\caption{Zoom of Fig.~\protect\figref{Overlap2} in vicinity of $a_2$
  overlap.
The angle gap at $b_3$ is $0.8^\circ$, and the gap at $b_2$ is $2.8^\circ$.}
\figlab{Overlap3}
\end{minipage}%
\end{figure}
Explict coordinates for the vertices of $\P$ are given in Sec.~\secref{Coordinates} of the Appendix.

One can identify two features of the polyhedron just described that
led to overlap: low curvature vertices (to restrict freedom) and obtuse face angles (at $a_1$ and $a_2$)
(to create ``overhang").
Both seem necessary ingredients.
Here I pursue
pursue excluding obtuse angles:

\begin{theorem}
If $\P$ is nonobtusely triangulated, then for every face $B$,
$N_v(B)$ has a petal unfolding.
\thmlab{NonObtuse}
\end{theorem}

A \emph{nonobtuse triangle} is one whose angles are each $\le \pi/2$.
It is known that any polygon of $n$ vertices has a nonobtuse
triangulation by $O(n)$ triangles, which can be found in
$O(n \log^2 n)$ time~\cite{bmr-lsntp-95}.
Open Problem~22.6~\cite[p.~332]{do-gfalop-07} asked whether every
nonobtusely triangulated convex polyhedron has an edge unfolding.
One can view Theorem~\thmref{NonObtuse} as a (very small) advance on
this problem.

The nonobtuseness of the triangles permits identifying
smaller \emph{diamond} regions $D_i$ inside the altitude regions $R_i$ used
in Sec.~\secref{ToplessPetalUnfolding},
such that $D_i$ necessarily contains the $A$-fan triangles, regardless of how
they are grouped.
See Fig.~\figref{Diamond}(a).
\begin{figure}[htbp]
\centering
\includegraphics[width=\linewidth]{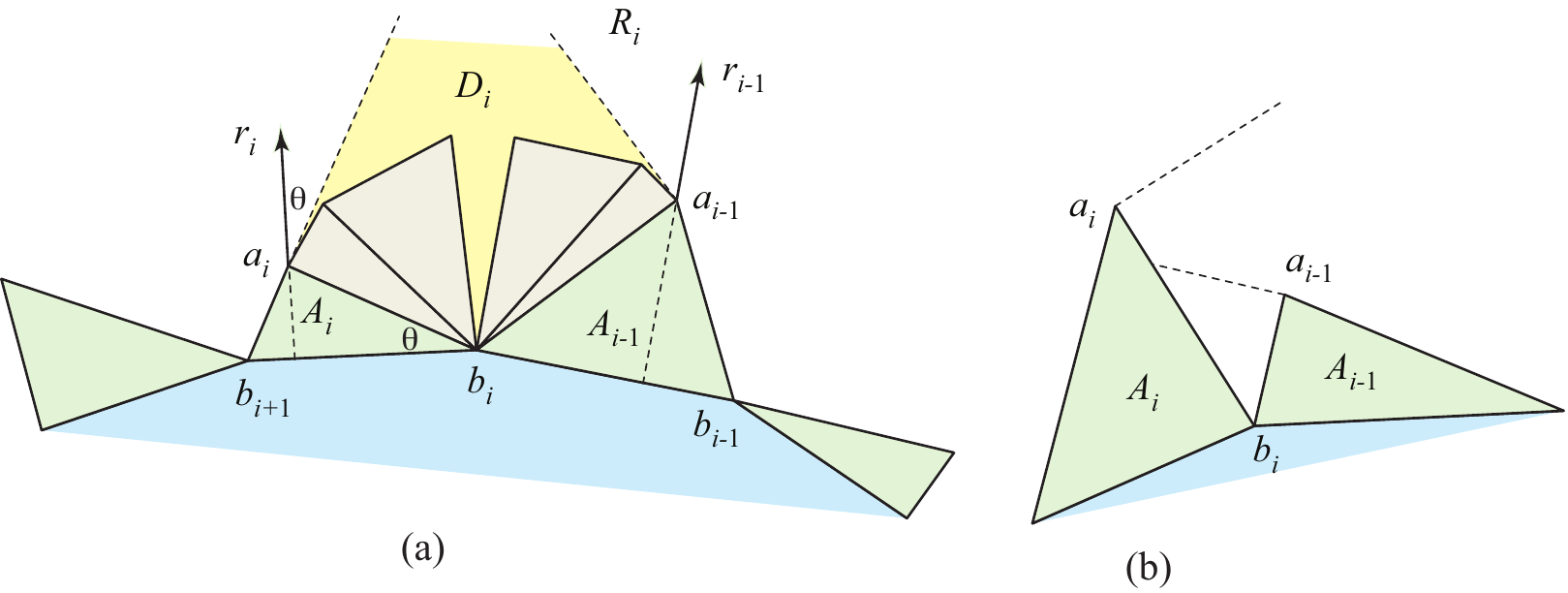}
\caption{(a)~$D_i \subset R_i$.
(b)~Perpendiculars cannot hit $A_i$ or $A_{i-1}$.}
\figlab{Diamond}
\end{figure}

A little more analysis leads to a petal unfolding of a (very special) class of prismatoids:
\begin{cor}
Let $\P$ be a triangular prismatoid all of whose faces, except possibly the base
$B$, are nonobtuse triangles, and the base is a (possibly obtuse)
triangle.  
Then every petal unfolding of $\P$ does not
overlap.
\corlab{NOPrismatoid}
\end{cor}
\begin{proof}
See Sec.~\secref{NOPrismatoid} in Appendix.
\end{proof}
Fig.~\figref{VTop} shows one illustration from the proof, which defines another region $V_i \supset R_i$
which does not overlap the adjacent diamonds $D_{i-1}$ and $D_{i+1}$, and into which it is safe to unfold the top $A$.
\begin{figure}[htbp]
\centering
\includegraphics[width=\linewidth]{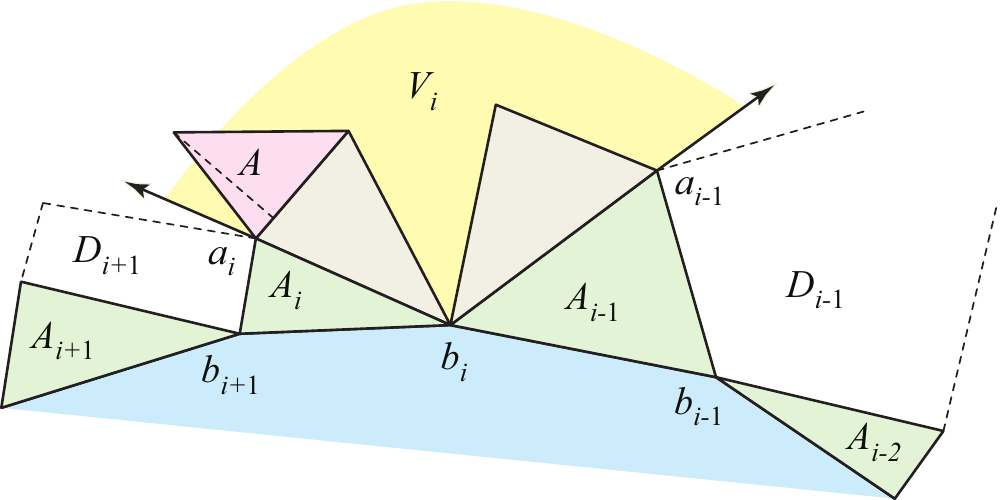}
\caption{The top $A$ of the prismatoid remains inside $V_i$.}
\figlab{VTop}
\end{figure}

\section{Unfolding Convex Patches}
\seclab{ConvexPatches}
I believe that unfolding convex patches may be a fruitful line of investigation.
For example, notice that the edges cut in a petal unfolding of a topless prismatoid
or of vertex-neighbhood of a face, form a disconnected spanning forest rather than a single spanning
tree.  One might ask:  Does every convex patch have an edge unfolding via single
spanning cut tree?
The answer is {\sc no}, already provided by the banded hexagon example
in
Fig.~\figref{BandedHex3D}.  For such a tree can only touch the boundary at one
vertex (otherwise it would lead to more than one piece), 
and then it is easy to run through the few possible spanning trees and show
they all overlap.

The term \emph{zipper unfolding} was introduced in~\cite{lddss-zupc-10} for
a nonoverlapping unfolding of a convex polyhedron achieved via Hamiltonian cut path.
They studied zipper edge-paths, following edges of the polyhedron, but
raised 
the interesting question of whether or not every convex
polyhedron has a zipper path, not constrained to follow edges, 
that leads to a nonoverlapping unfolding.
This is a special case of Open Problem~22.3
in~\cite[p.~321]{do-gfalop-07}
and still seems difficult to resolve.

Given the focus of this work, it is natural to specialize this
question further, to ask if every convex patch has a zipper unfolding, using
arbitrary cuts.
I believe the answer is negative:
a version of the banded hexagon shown in See Fig.~\figref{ThinApron}
has no zipper unfolding.
My argument for this is long and seems difficult to formalize, so I
leave the claim as a conjecture that, with effort,
the proof could be formalized.
\begin{figure}[htbp]
\centering
\includegraphics[width=0.75\linewidth]{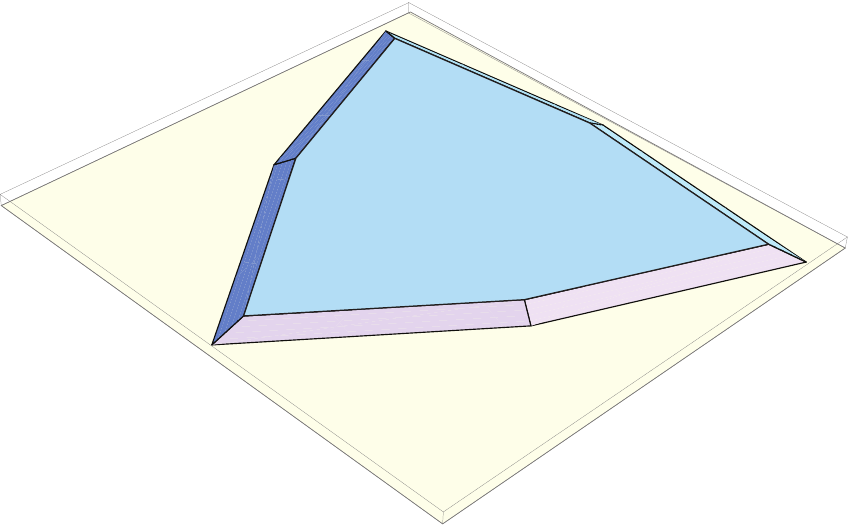}
\caption{The banded hexagon with a thin band.}
\figlab{ThinApron}
\end{figure}



\small
\bibliographystyle{alpha}
\bibliography{/Users/orourke/bib/geom/geom}
\normalsize

\clearpage
\section{Appendix}

\begin{figure}[htbp]
\centering
\includegraphics[width=\linewidth]{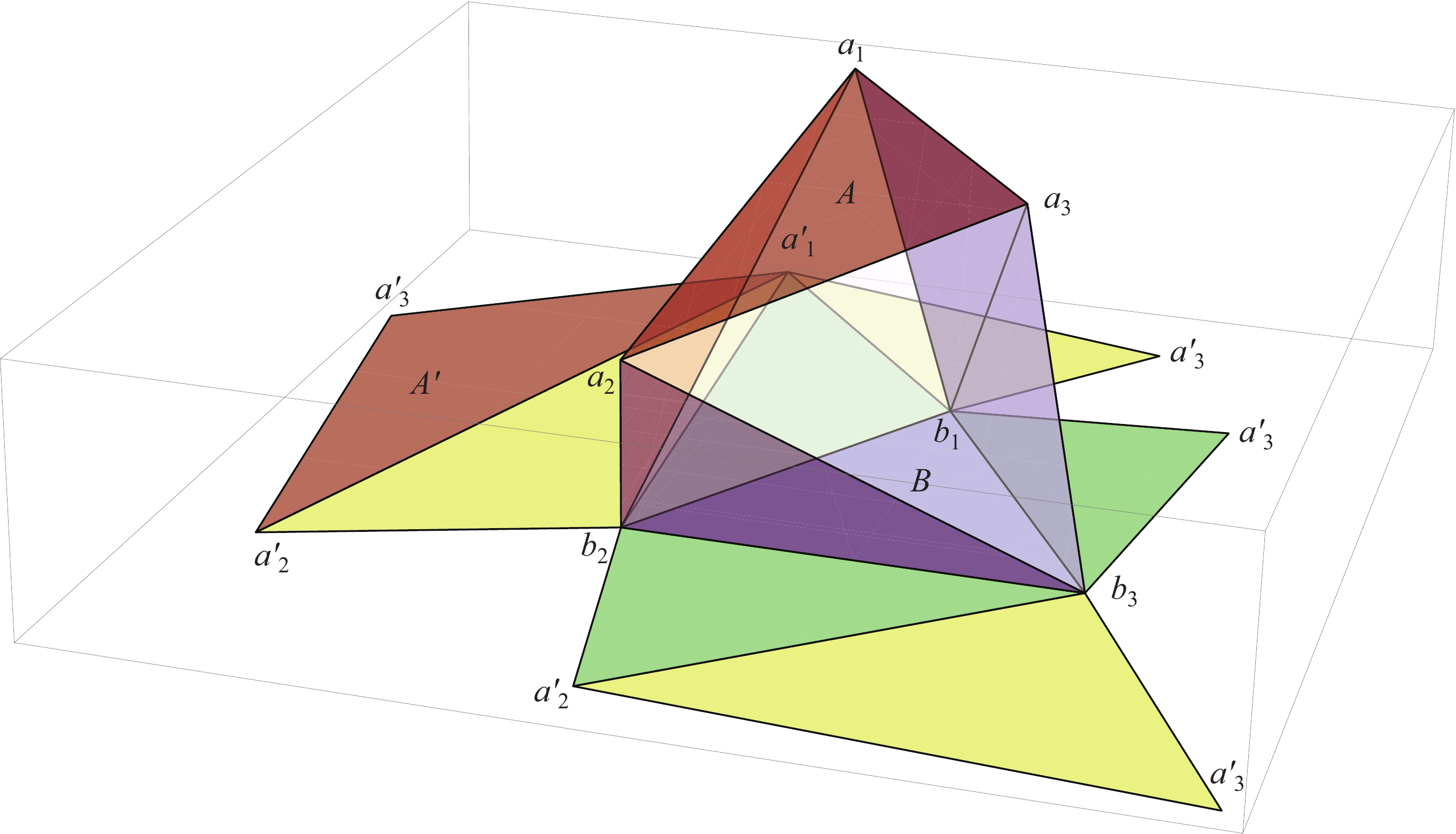}
\caption{A triangular prismatoid (top and bottom both triangles),
and one petal unfolding.  
The base $B$-triangles are green; the
top $A$-triangles are yellow.}
\figlab{TriPrismatoidExample}
\end{figure}

\begin{figure}[htbp]
\centering
\includegraphics[width=\linewidth]{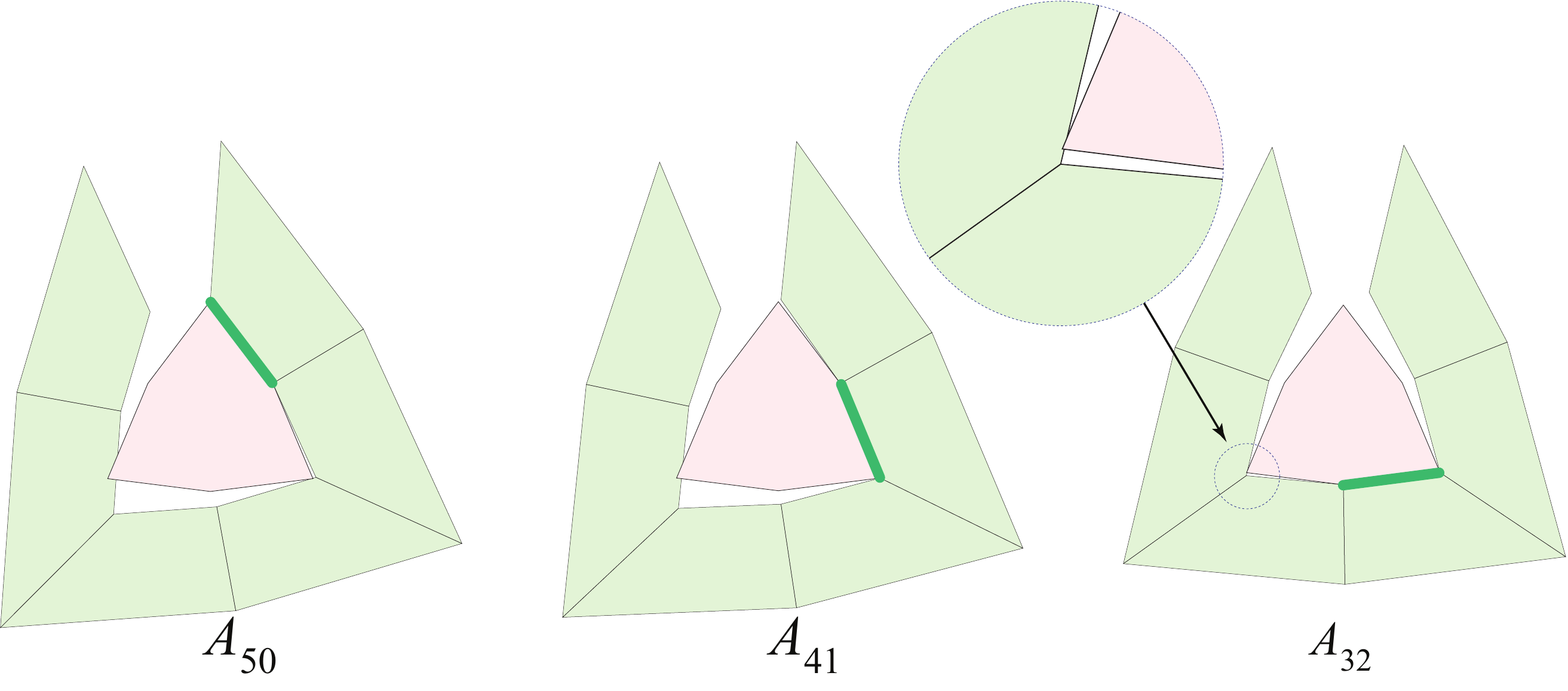}
\caption{Apex cuts: each leads to overlap. The highlighted edge is not cut.}
\figlab{ApexCuts}
\end{figure}
\begin{figure}[htbp]
\centering
\includegraphics[width=\linewidth]{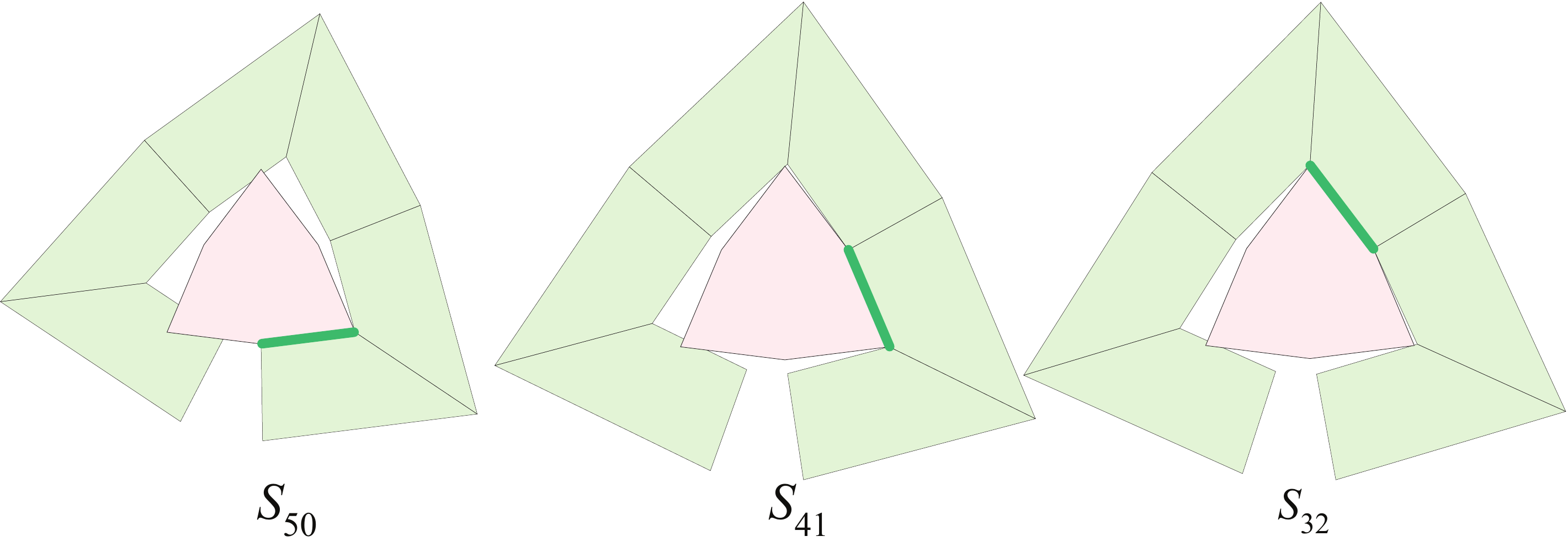}
\caption{Side cuts: each leads to overlap.}
\figlab{SideCuts}
\end{figure}
\begin{figure}[htbp]
\centering
\includegraphics[width=\linewidth]{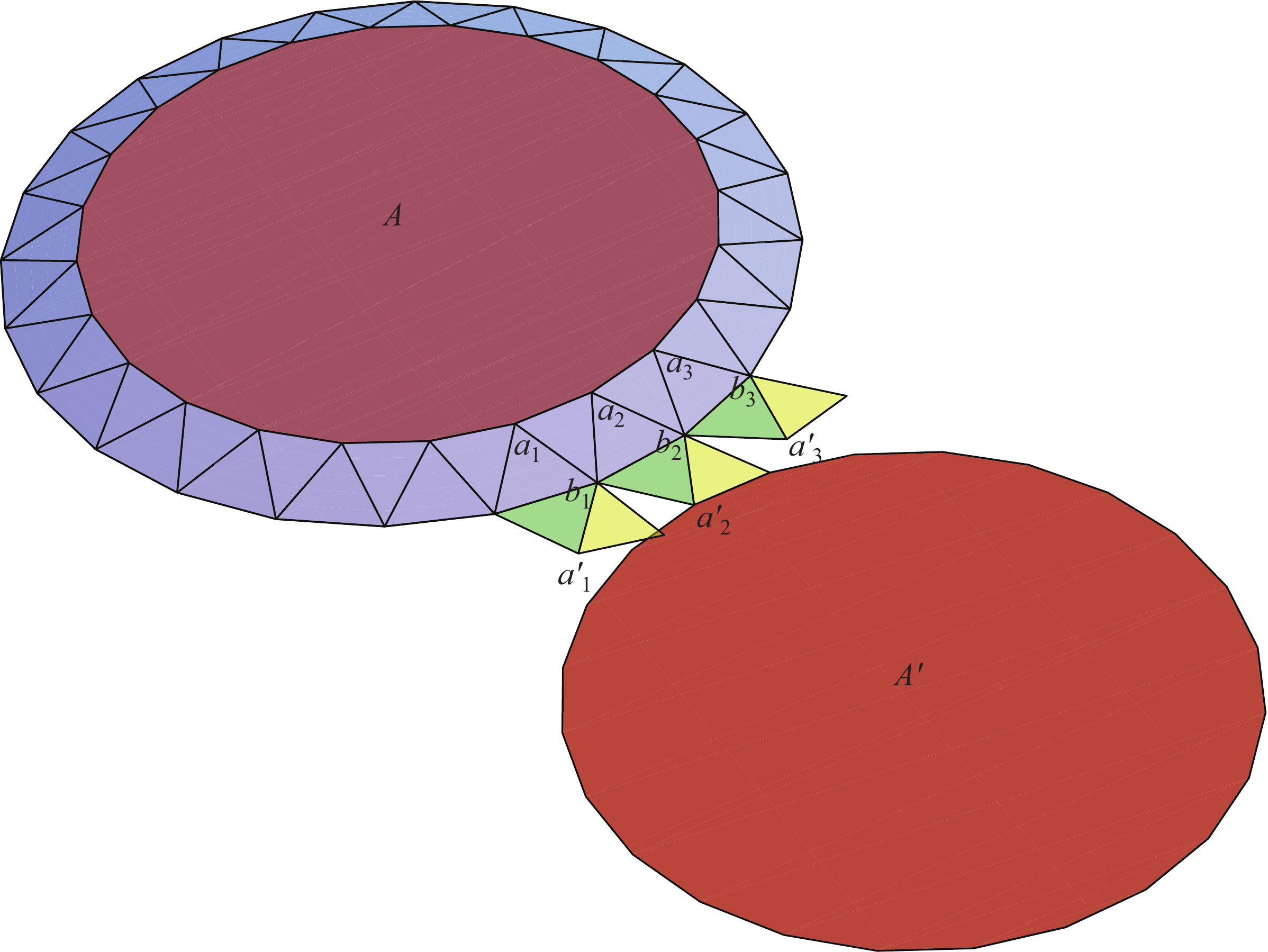}
\caption{A drum-like prismatoid that results in overlap with consistent ccw rotation
of the (yellow) $A$-triangles.  Here the point $a'_1$ overlaps the unfolded top $A'$.
This overlap can be removed easily, by rotating the $A$-triangle $\triangle a_1 a_2 b_1$ cw
rather than ccw.}
\figlab{DrumOverlap}
\end{figure}
\begin{figure}[htbp]
\centering
\includegraphics[width=0.75\linewidth]{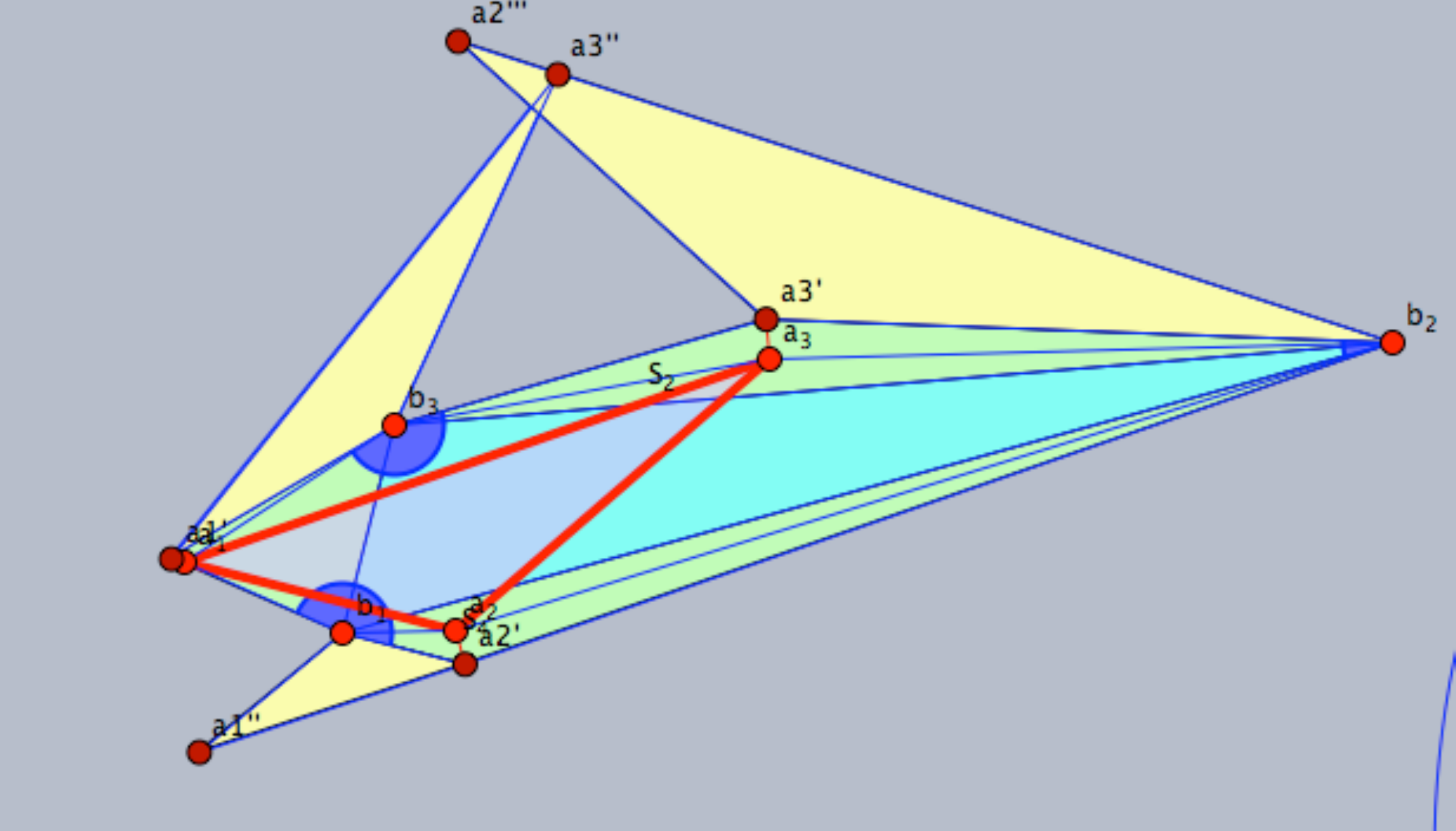}
\caption{An overhead view of a nearly flat, topless triangular prismatoid.
$A$-triangles $\triangle a_2 a_3 b_2$ and $\triangle a_3 a_1 b_3$ are both
rotated ccw, about $b_2$ and $b_3$ respectively.
[Figure created in Cinderella.]
}
\figlab{WingsCrossCcw}
\end{figure}

\subsection{Proof of Lemma~\protect\lemref{AltitudeRays}}
\seclab{ProofAltitudeRays}
\setcounter{theorem}{0}
\begin{lemma}
No pair of altitude rays cross in the base plane, and so they define a partition of that plane exterior
to the base unfolding.
\end{lemma}
\begin{proof}
Consider three consecutive $B$ vertices of the prismatoid $\P$, $(b_1,b_2,b_3)$
supporting two base triangles, $B_1=\triangle b_1 b_2 a_1$ and $B_2=\triangle b_2 b_3 a_2$.
We will show that $r_1$ and $r_2$ cannot cross.
Let $\b_1 = \angle b_1 b_2 a_1$ and $\b_2 = \angle b_3 b_2 a_2$ be the two angles
of the base triangles incident to $b_2$.
(We use $a_2$ for the apex of $B_2$ for simplicity, although there could be intervening $A$ vertices between $a_1$ and $a_2$.)
We consider three cases, distinguishing acute and obtuse $\b_i$ angles.
\begin{figure}[htbp]
\centering
\includegraphics[width=\linewidth]{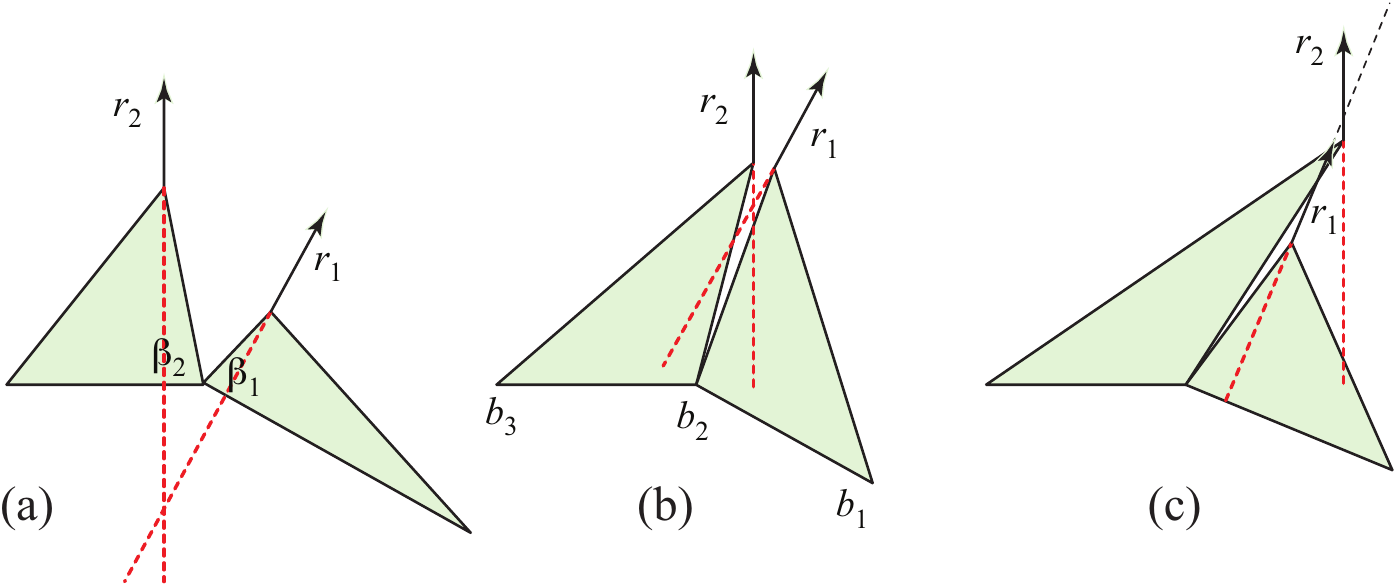}
\caption{Only in case~(c) could ray $r_1$ cross $r_2$.}
\figlab{ObtuseAcuteRays}
\end{figure}

If both $\b_1$ and $\b_2$ are acute, then the altitudes of $B_1$ and $B_2$ lie on the
base edges $b_1 b_2$ and $b_2 b_3$ respectively, and the lines containing the rays
cross behind the rays, as in Fig.~\figref{ObtuseAcuteRays}(a).
Similarly, if both $\b_1$ and $\b_2$ are obtuse, again the ray lines cross behind
the rays, this time exterior to $B$, as in~(b) of the figure.
Only when one angle is obtuse and the other acute could the rays possibly
cross.  Without loss of generality, let $\b_2$ be obtuse and $\b_1$ acute,
as in~(c) of the figure.  We now concentrate on this case.

Let $H_i$ be the vertical plane containing the altitude of $B'_i$.
This plane 
includes both 
the unfolded $a'_i$ on the $B$-plane and the vertex $a_i$ on the $A$-plane,
because $a'_i$ is the image of $a_i$ rotated about the base edge $b_i b_{i+1}$
to which the altitude of $B_i$ is perpendicular.
See Fig.~\figref{AltitudeRays3D}.
\begin{figure}[htbp]
\centering
\includegraphics[width=\linewidth]{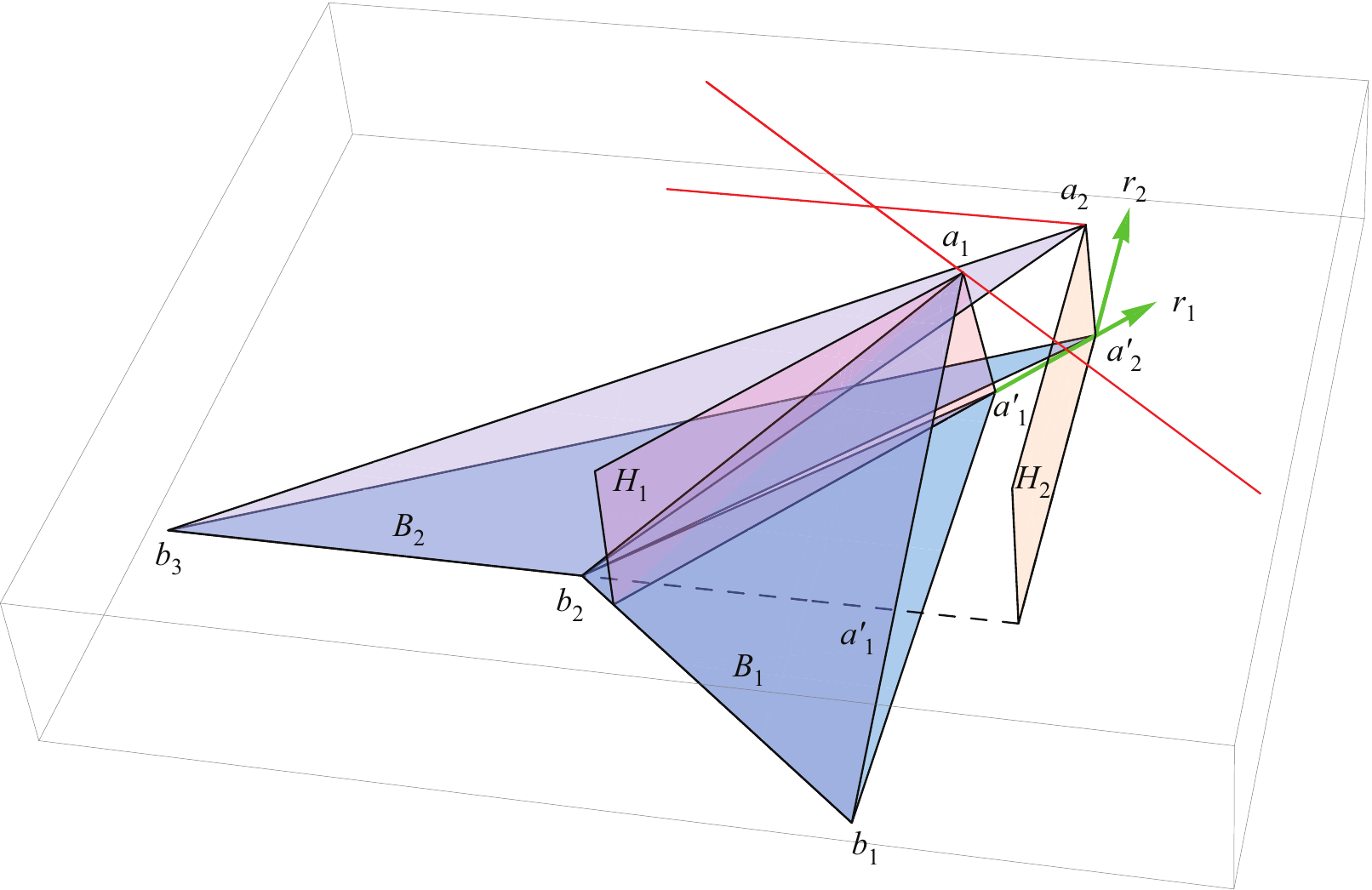}
\caption{The conditions of this case violate the convexity of $\P$: 
$a_1$ must be right of $H_2$ so that $a_2$ is inside the plane determined by $B_1$.}
\figlab{AltitudeRays3D}
\end{figure}
The $B_i$ triangles of $\P$ cut the $A$-plane in lines parallel
to their base edges $b_i b_{i+1}$, and the top $A$ must fall inside the halfplanes 
on the $A$-plane bounded by these lines.
Examination of the figure shows that this requires $a_1$ to lie on
the $A$-plane right of $H_2$ in the figure.
But $a'_1$ is necessarily initially left of $H_2$ if $r_1$ is to cross $r_2$,
and the rotation of $a'_1$ from the $B$-plane up to the $A$-plane moves it
only further left of $H_2$.
Thus this last case violates the convexity of $\P$,
and we have established the lemma for adjacent altitude rays $r_1, r_2$.

(We have shown in the figure $B_1$ and $B_2$ both making an angle less than $\pi/2$ with the base plane,
but the argument is not altered if either of those angles exceed $\pi/2$: still the rotation of $a_i$ 
down to $a'_i$ occurs in the altitude $H_i$ plane.)

Now consider nonadjacent rays, say $r_1$ and $r_i$, based on base triangles
$B_1$ and $B_i$.  
Extend the edges of those triangles in the $B$-plane until they meet at point $\overline{b}$,
and form new triangles $\overline{B_1}=\triangle b_1 \overline{b} a_1$ 
and $\overline{B_i}=\triangle \overline{b} b_{i+1} a_i$ sharing $\overline{b}$.
(Again we use $a_i$ for the apex of $B_i$ without implying there are exactly $i-1$ $A$-vertices between $a_1$ and $a_i$.)
Notice these triangles are still apexed at $a_1$ and $a_i$ respectively, as the planes containing
$B_1$ and $B_i$ support $A$ at these two points.
Define $\overline{\P}$ as the convex hull of $\P \cup \overline{b}$.
In $\overline{\P}$, the altitudes of the new base triangles 
$\overline{B_1}$ and $\overline{B_i}$ are exactly the same as the altitudes
of the original $B_1$ and $B_i$, because their base edges have been extended while
retaining their apexes on $A$.
So the rays $r_1$ and $r_i$ have not changed in the base plane, and we can reapply the argument for adjacent rays.
\end{proof}

\subsection{Proof of Lemma~\protect\lemref{zto0}}

\begin{lemma}
Let $\P(z)$ be a prismatoid with height $z$.
Then the combinatorial structure of $\P(z)$ is independent of $z$,
i.e., raising or lowering $A$ above $B$ retains the convex hull structure.
\end{lemma}
\begin{proof}
Let $B_1=\triangle b_1 b_2 a(z)$ be a $B$-triangle for some $z>0$.  
(The argument is the same for an $A$-triangle by inverting $\P$.)
Let $L(z)$ be the line in the $A$-plane parallel to $b_1 b_2$ through $a(z)$,
i.e., $L(z)$ is the intersection of the plane containing $B_1$ with the $A$-plane.
Then $L(z)$ is a line of support for $A(z)$ in the $A$-plane.
As $z$ varies, this line remains parallel to $b_1 b_2$, and because $A(z)$ merely translates with $z$
(it does not rotate), $L(z)$ remains a line of support to $A(z)$.  Thus the plane containing $B_1(z)$
supports $A(z)$, and of course it supports $B$ because $b_1 b_2$ does not move.
Therefore, $B_1(z)$ remains a face of $\P(z)$ for all $z>0$.
\end{proof}

\subsection{Proof of Lemma~\protect\lemref{AltitudeRaysUnf}}
\seclab{AltitudeRaysUnf}
\begin{lemma}
Let $\P(z)$ be a prismatoid with height $z$, and $BU(z)$ its base unfolding.
Then the apex $a'_j(z)$ of each $B'_i(z)$ triangle $\triangle b_i b_{i+1} a'_j(z)$
in $BU(z)$ lies on the
fixed line containing the altitude of $B'_i(z)$.
\end{lemma}
\begin{proof}
Recall that $B'_i$ is produced by rotating $B_i$ about its base edge $b_i b_{i+1}$.
Thus every point on a line perpendicular to $b_i b_{i+1}$ lying within the plane
of $B_i$ unfolds to that line rotated to the base plane.
Because $a_j(z)$ lies on such a line containing $B_i$'s altitude, $a'_j(z)$ 
is on the line containing the altitude to $B'_i$.
\end{proof}

\begin{figure}[htbp]
\centering
\includegraphics[width=\linewidth]{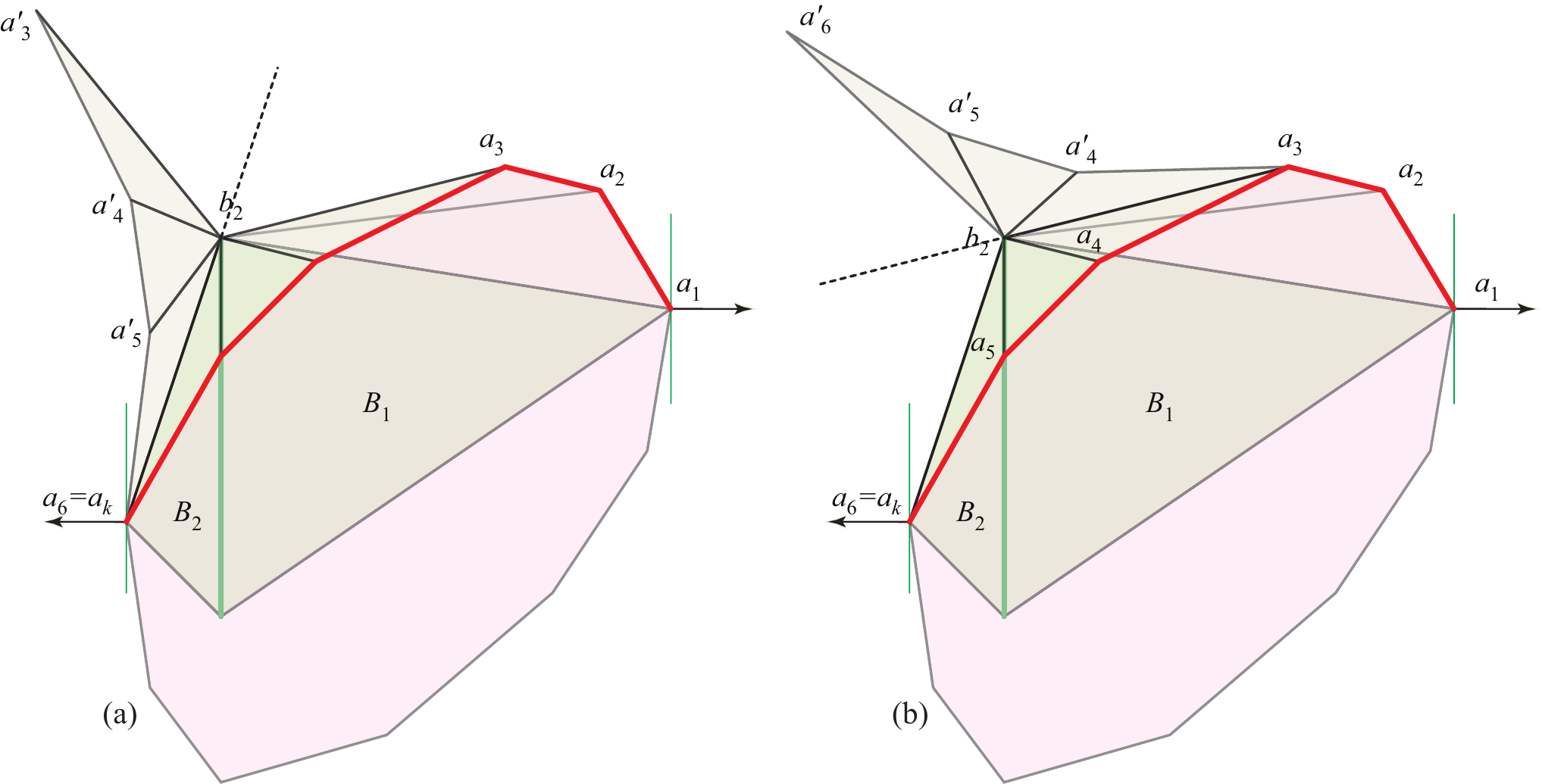}
\caption{Case~2 gone bad: the chain $(a'_4, a'_5, a'_6)$ leaves $R$ as it
crosses $r_1$.  The overlap in Fig.~\protect\figref{WingsCrossCcw} can also be 
understood as 
caused by an unsafe flip.}
\figlab{FlatFlippingCase2Bad}
\end{figure}

\subsection{Proof of Lemma~\protect\lemref{Flip}}
\seclab{Flip}

\begin{lemma}
Let $b_2$ have tangents $a_s$ and $a_t$ to $A$.
Then either reflecting the enclosed up-faces across the left tangent,
or across the right tangent, is ``safe" in the sense that no points of
a flipped triangle falls outside
the rays $r_1$ or $r_k$.
\end{lemma}
\begin{proof}
The rays $r_1$ and $r_k$ are in general below and turned beyond
(ccw and cw respectively)
the tangency points $a_s$ and $a_t$, but at their ``highest" they are as
illustrated in Fig.~\figref{FlipLemma}.
If reflecting $a_s$ to $a'_s$ is not safe as illustrated, 
then the perpendicular at $a_t$ must
hit $b_2 a_s$.  Because it makes an angle $\b$ there with $a_t a'_t$, the alternate
reflection is safe. 
\begin{figure}[htbp]
\centering
\includegraphics[width=0.75\linewidth]{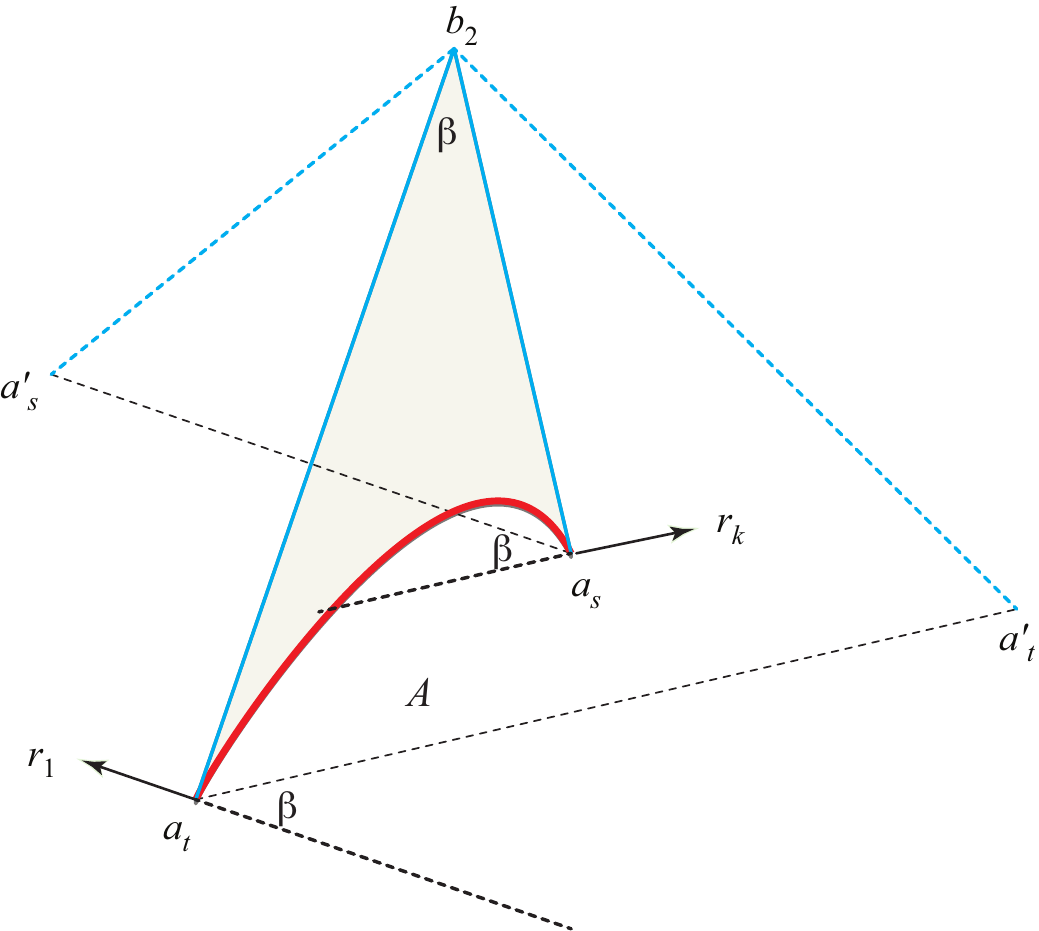}
\caption{One of the two reflections must remain above the rays $r_1$ or $r_k$.}
\figlab{FlipLemma}
\end{figure}
\end{proof}

\begin{figure}[htbp]
\centering
\includegraphics[width=\linewidth]{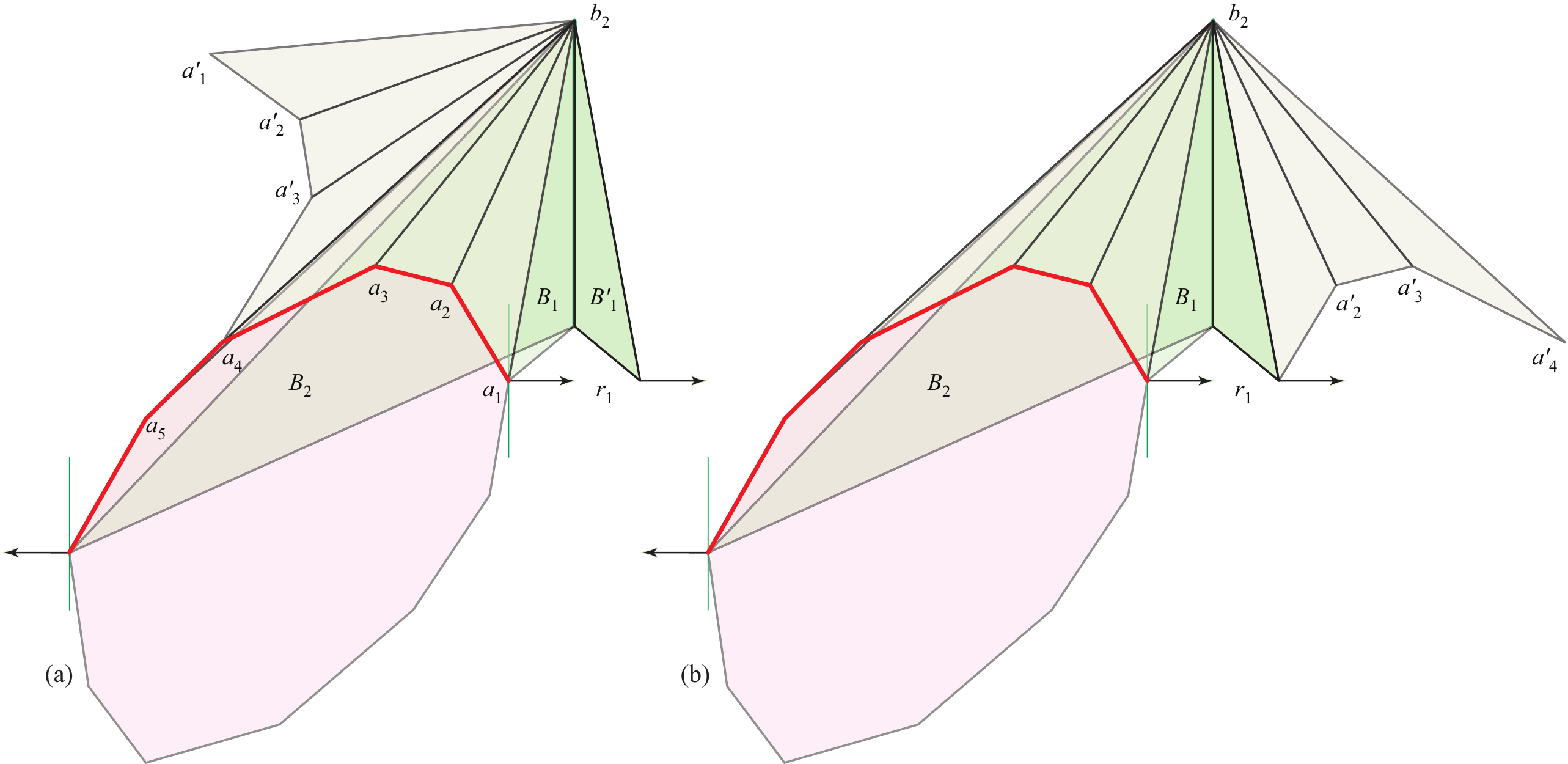}
\caption{Case~2b.  Here $B_1$ is an up-face.  (a)~Flip across the left tangent.
(b)~Rather than flip the up-$A$-faces across the right
tangent , those faces are flipped while attached to $B_1$---i.e., we treat $B_1$ as joined to those
up-$A$-faces.
}
\figlab{FlatFlippingCase3}
\end{figure}

\subsection{Proof of Lemma~\protect\lemref{AngleMonotonicity}}
\seclab{AngleMonotonicity}
\begin{lemma}
Let $\triangle b ,a_1(z) ,a_2(z)$ be an $A$-triangle, with angles
$\a_1(z)$ and $\a_2(z)$ at $a_1(z)$ and $a_2(z)$ respectively.
Then $\a_1(z)$ and $\a_2(z)$ are monotonic from their $z=0$ values
toward $\pi/2$ as $z \to \infty$.
\end{lemma}
\begin{proof}
With loss of generality, let $b=(0,0,0)$, $a_1(z)=(1,0,z)$, and $a_2=(1+x,y,z)$,
with $y>0$.  If $x>0$, then $\a_1(z) > \pi/2$ (obtuse), and if $x \le 0$, then
$\a_1(z) < \pi/2$ (acute).  By symmetry, we need only prove the claim for
$\a_1(z)$.

The dot-product $(a_1(z) -b) \cdot (a_2(z) - a_1(z))$ determines 
either $\cos (\a_1(z))$ or $\cos (\pi-\a_1(z))$, depending on whether or
not $\a_1(z)$ is acute or not.
Direct computation leads to
$$
\cos (\;) = \frac{x} {  \sqrt{x^2+y^2} \sqrt{ 1 + z^2} }
$$
whose derivative with respect to $z$ is
$$
\frac{-xz} {  \sqrt{x^2+y^2} { (1 + z^2)^{3/2}} } \;.
$$
Because $z>0$, the sign of the derivative is entirely determined by the sign of $x$.
For $\a_1$ obtuse, $x>0$, the derivative is negative, which corresponds to 
decreasing $\a_1(z)$, and when $x<0$ and $\a_1$ is acute, the derivative is
positive corresponding to increasing $\a_1(z)$.
Thus the claim of the lemma is established.
\end{proof}

\begin{figure}[htbp]
\centering
\includegraphics[width=0.75\linewidth]{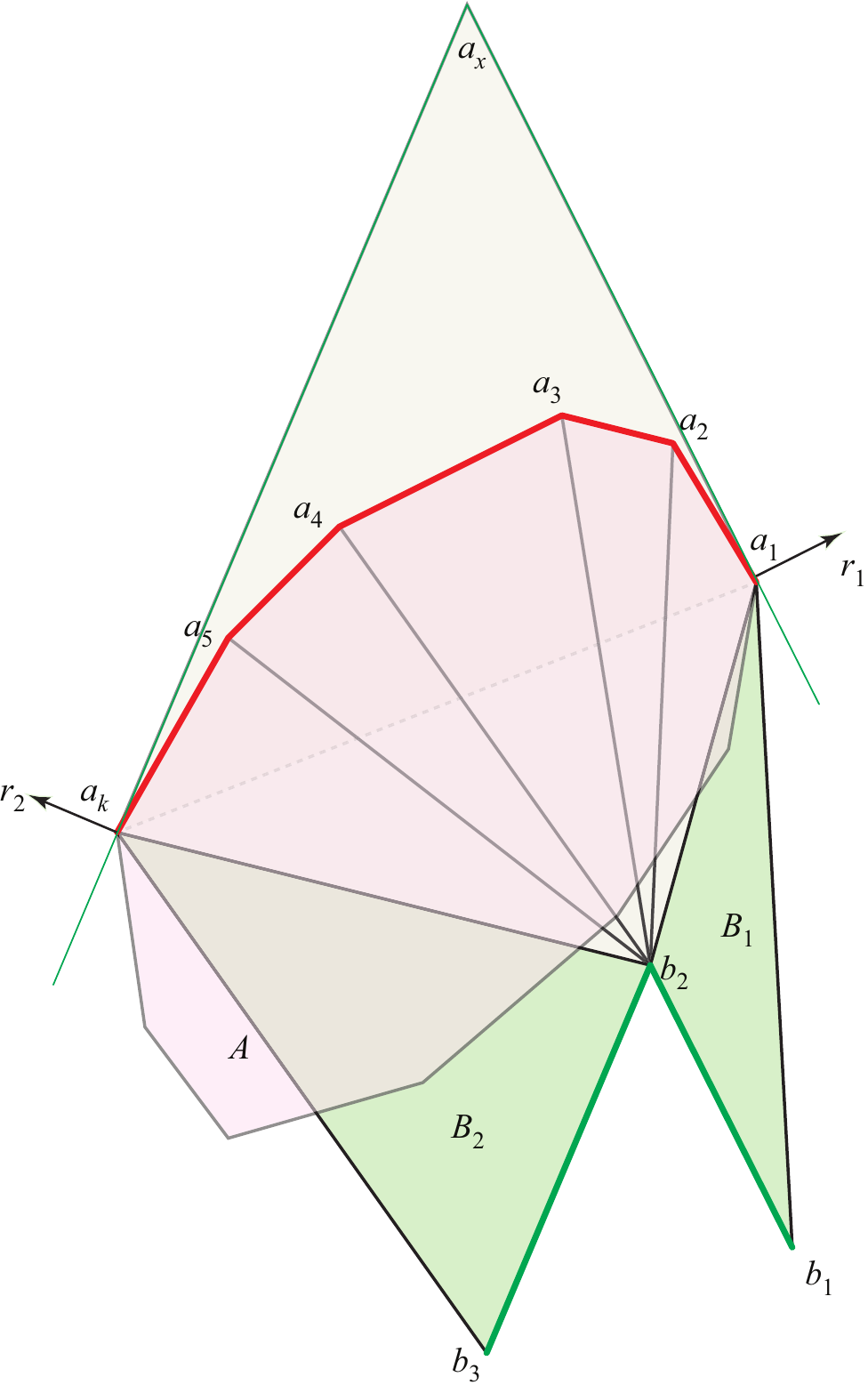}
\caption{Enclosing a convex chain with a triangle $\triangle a_1 a_x a_k$,
where $a_x$ is the intersection of lines of support at $a_1$ and $a_k$
parallel to $b_1 b_2$ and $b_2 b_3$ respectively.}
\figlab{AllConvexEncTri}
\end{figure}

\subsection{Proof of Lemma~\protect\lemref{CRz}}
\seclab{CRz}

Here we will need two important facts about the unfolded $a$-chain:
\begin{enumerate}
\item
Let $\a_j$ be the angle of the chain at $a_j$, i.e., the sum
of the two incident triangle angles, $\angle b_2 a_j a_{j-1} + \angle b_2 a_j a_{j+1} $.
If $\a_j$ is convex for $z=0$, it remains convex for all $z$; and similarly reflex
remains reflex, and a sum of $\pi$ remains independent of $z$.
\item $\a_j(z)$ is monotonic with respect to $z$, approaching $\pi$ as $z \to \infty$
from above (if initially reflex) or below (if initially convex).
\end{enumerate}
The essence of why Fact~1 holds is in
Fig.~\figref{SumPi}.
\begin{figure}[htbp]
\centering
\includegraphics[width=\linewidth]{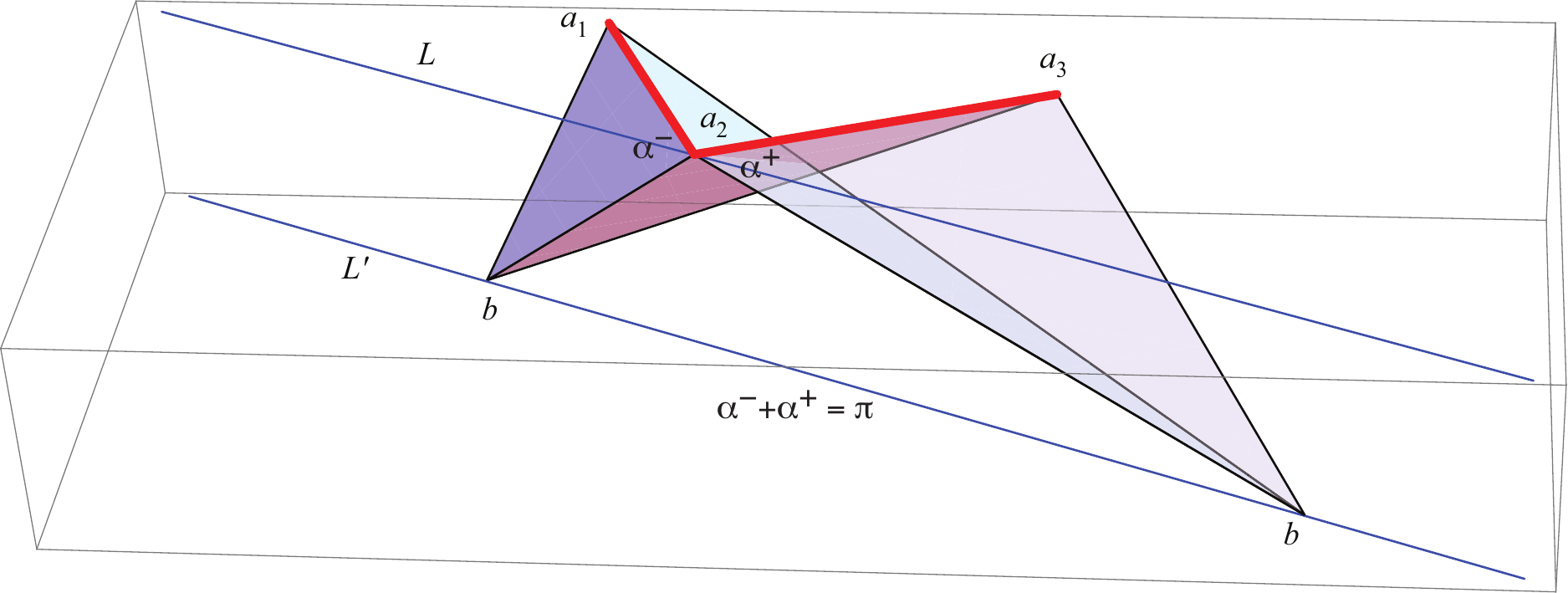}
\caption{The locus of positions $b$ for which $\a^- + \a^+ = \pi$}
\figlab{SumPi}
\end{figure}
See~\cite{o-taspi-12} for proofs.
Fact~2 can be established by superimposing neighborhoods of 
$a_j$ for two different $z$-values $z_1 < z_2$,
and noting, for reflex $\a_j$, the $z_2$-neighborhood
is nested in that for $z_1$, and consequently there is a larger
curvature $\kap_{a_j}(z_2) >\kap_{a_j}(z_1)$.

\seclab{CRz}
\begin{lemma}
If the $a$-chain consists of a convex and a reflex section, and the safe flip 
(by Lemma~\lemref{Flip}) is to
a side with a down-face ($B_2$ in the figure), then $AF'(z) \subset R(z)$:
the $A$-fan unfolds within the altitude region for all $z$.
\end{lemma}
\begin{proof}
Let $a_s$ and $a_t$ be the vertices of the $a$-chain so that lines 
continaining $b_2 a_s$ and
$b_2 a_t$ are supporting tangents to $A$ at $a_s$ and $a_t$.
Thus $(a_1,\ldots,a_s)$ represents a convex portion of the $a$-chain,
$(a_s,\ldots,a_t)$ the reflex portion, and
$(a_t,\ldots,a_k)$ a convex portion.
We first assume $a_s=a_1$ so we have only a convex and a reflex section,
as illustrated in Fig.~\figref{FlippingCR}.
We also first assume that both $B_1$ and $B_2$ are down-faces and so do
not require flipping.
We analyze this case by mixing the convex and reflex approaches in earlier, easier cases not detailed here (but see Fig.~\figref{AllConvexEncTri}).

For the reflex chain, we connect $a_s=a_1$ to $a_t$ to form
a triangle $A_{st}=\triangle a_s b_2 a_t$ that encloses the reflex chain.
For the convex chain $(a_t,\ldots,a_k)$ we intersect the line $L_{23}$
parallel to $b_2 b_3$ through $a_k$ (just as in the all-convex case not detailed),
and intersect it with the line containing $b_2 a_t$.  Let that intersection point
be $a_x$.  Then the triangle $A_x=\triangle b_2 a_x a_k$ encloses the
convex chain.  Under the assumption that $B_1$ is a down-face, then 
$A_x$ encloses all down-faces, and does not need flipping.
$A_{st}$ does flip, and let us assume the safe flip is across $b_2 a_t$,
flipping $a_s$ to $a'_s$, with $A'_{st}$ the reflected triangle.

Vertex $a_k(z)$ rides out $r_2$.
By  construction, $a_x(z) a_k(z) \perp r_2$, as $a_x$ was defined by
$L_{23}$ parallel to $r_2$.  Because $|a_x(z) a_k(z)|=|a_x a_k|$,
$a_x(z)$ rides out along a line parallel $L_x$ to $r_2$, so $A_x(z) \subset R(z)$.

Now the curvature $\kap(z)$ at $b_2$, i.e., the angle gap in the unfolding, varies
in a possibly complex way, but it remains positive at all times, because clearly
$\P(z)$ is not flat at $b_2$ for any $z$. Thus $b_2 a'_1(z)$ is rotated ccw from
$b_2 a'_1(z)$.  
It remains to show that $b_2 a'_1(z)$ cannot cross $r_2$.

By Fact~1 above, the convex angle at $a_x$ remains convex at $a_x(z)$,
and therefore $a_t(z)$ cannot cross $L_x$ let alone $r_2$.
Again by Fact~1, the reflex chain $(a_1,\ldots,a_t)$ remains a reflex chain
with increasing $z$, and so is contained inside $A'_{st}(z)$.
This reflex chain straightens, approaching the segment $a_t(z) a'_1(z) $.

Because that chain is reflex, the only way that $A'_{st}$ can cross $r_2$ is
for the segment $a_t(z) a'_1(z) $ to cross, i.e., for $a'_1(z)$ to cross.
Notice this requires a highly reflex angle $\a_t(z)=\angle a'_1(z),a_t(z),a_x(z)$, at least $3\pi/2$ in fact,
in order to cross over the line $L_x$.
Now we have no control over the initial value of $\a_t$, but we know that
the flip was safe, so initially $a'_1$ is inside $r_2$.  If $\a_t$ is convex, then
$\a_t(z)$ remains convex and $a'_1(z)$ cannot cross $r_2$.
So assume $\a_t$ is initially reflex (as illustrated in Fig.~\figref{FlippingCR}).
Then by Fact~2, it decreases monotonically toward $\pi$ as $z$ increases.
Because it decreases, and needs to be at least $3\pi/2$ to cross $r_2$,
it must have started out at least $3\pi/2$.
Now we argue that this is impossible, as the other flip would have been chosen.

As Fig.~\figref{FlipLemma2} shows, if $\a_t > 3\pi/2$, then the reflection
$a_t a'_1$ is already more than $\pi/2$ ccw of $b_2 a_t$, which
marks it as an unsafe flip. We would instead have flipped the reflex portion
across $b_2 a_1=a_s$.  And indeed the flip in Fig.~\figref{FlippingCR} 
would not have been chosen because it is potentially unsafe (but does not
in this case actually place $a'_1$ on the wrong side of $r_2$).
\begin{figure}[htbp]
\centering
\includegraphics[width=0.75\linewidth]{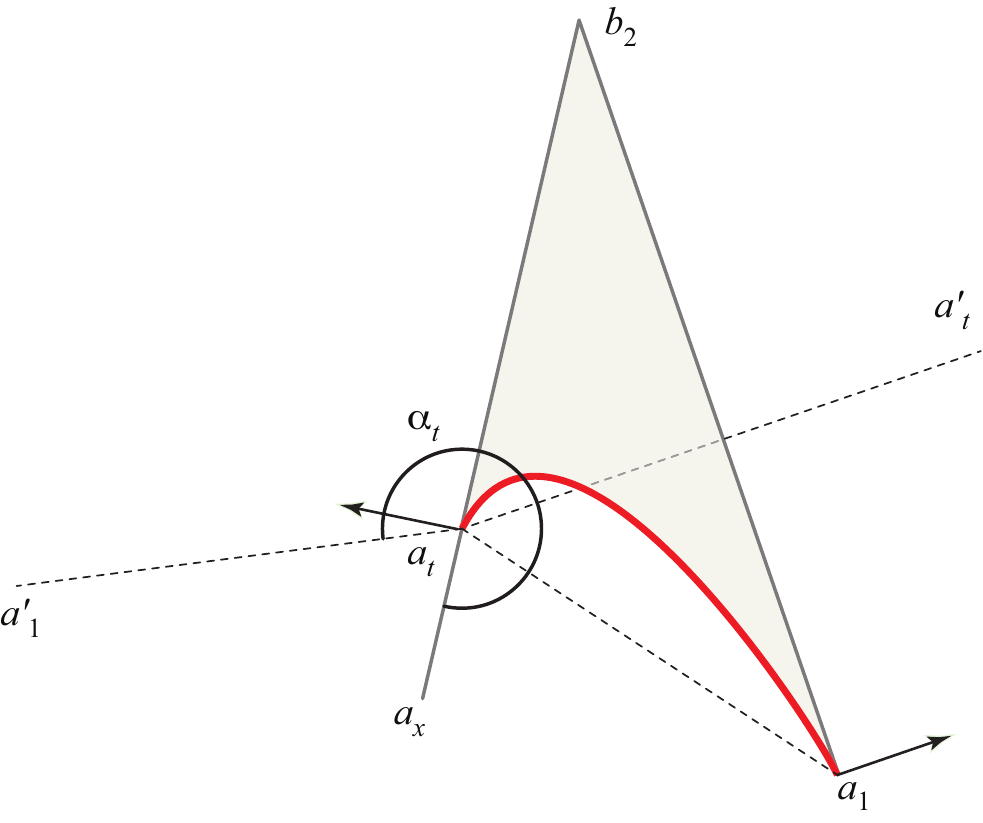}
\caption{In order for $\a_t > 3\pi/2$, $a_t a_1$ must make an angle more than $\pi/2$
with $b_2 a_t$.}
\figlab{FlipLemma2}
\end{figure}
\end{proof}

\subsection{Vertex-Neighborhood Counterexample Coordinates}
\seclab{Coordinates}
The coordinates of the nine vertices comprising $\P$ in Fig.~\figref{CatsEye3D1} are shown in the table
below,
with $\{a_2,b_3,c_2,p_3\}$ each reflections of $\{a_1,b_1,c_1,p_1\}$
with respect to the $x=0$ plane:
\begin{center}
\begin{tabular}{|c|c|}
\hline
Point & Coordinates\\
\hline
\hline
$b_2$ & $(0,0,0.2)$\\
$a_1,a_2$ & $(\pm 0.603496,0.0399127,0.2)$\\
$b_1,b_3$ & $(\pm 2,-0.1,0)$\\
$c_1,c_2$ & $(\pm 0.0124876,0.501659,0.2)$\\
$p_1,p_3$ & $(\pm 6.03626,-0.4,-0.6)$\\
\hline
\end{tabular}
\end{center}

\subsection{Proof of Corollary~\protect\corref{NOPrismatoid}}
\seclab{NOPrismatoid}
\begin{cor}
Let $\P$ be a triangular prismatoid all of whose faces, except possibly the base
$B$, are nonobtuse triangles, and the base is a (possibly obtuse)
triangle.  
Then every petal unfolding of $\P$ does not
overlap.
\end{cor}
\begin{proof}
We first let $B$ be an arbitrary convex polygon.
We define yet another region $V_i \supset R_i$ incident to $b_i$, bound by rays
from $b_i$ through $a_{i-1}$ and through $a_i$.
See Fig.~\figref{VTop}.
Note that these rays shoot at or above the adjacent diamonds $D_{i-1}$
and $D_{i+1}$,
and therefore miss $A_{i-2}$ and $A_{i+1}$.

Now we invoke the assumption that $B$ is a triangle: In that case,
those adjacent diamonds contain all the remaining $A$-triangles, because there
are
only three $b_i$ vertices: $b_1$ at which $V_1$ is incident, and
diamonds $D_2$ and $D_3$ to either side.
(Note there can only be altogether three $A$-triangles, one for each
edge of $A$.)
Now unfold the top $A$ of $\P$ attached to some $A$-triangle, without
loss
of generality a $A$-triangle incident to $b_1$.  Then because $A$ is
nonobtuse,
its altitude, and indeed all of $A$, projects into that edge shared
with a $A$-triangle $A_1$.  
Because the top of the $A$-triangle is inside $D_1$, we can see that
$A \subset V_i$, and we
have
protected $A$ from overlapping any other $A$-triangle or any $A_i$.
\end{proof}

It seems quite likely that this corollary still holds with $B$ an arbitrary
convex polygon, but, were the same proof idea followed,
it would require showing that $V_i$ does not
intersect
nonadjacent diamonds or more distant $A_j$ triangles.

\end{document}